\documentclass[journal]{IEEEtran}
\DeclareFixedFont{\auacc}{OT1}{phv}{b}{it}{18}   
\DeclareFixedFont{\newauacc}{OT1}{ptm}{b}{rm}{12}   
\usepackage{hyperref}
\usepackage{cite}
\usepackage{graphicx}
\usepackage[tight]{subfigure}
\usepackage{breakurl}
\usepackage[amsthm]{ntheorem}
\usepackage{verbatim}

\usepackage{amsmath}
\usepackage{amsfonts}
\usepackage{amssymb}
\usepackage[usenames]{color}
\usepackage{paralist}
\usepackage{ifthen}
\usepackage{tabularx}
\usepackage{multirow}
\usepackage{wrapfig}

\usepackage[dvipsnames]{xcolor}
\usepackage{colortbl}
\definecolor{mygray}{gray}{0.8}
\definecolor{mycol}{rgb}{1,0.9,0.9}
\definecolor{mcy}{rgb}{0.9,1,0.9}

\newboolean{arxiv}
\setboolean{arxiv}{false}

\DeclareMathSymbol{\R}{\mathord}{AMSb}{"52}
\DeclareMathSymbol{\C}{\mathord}{AMSb}{"43}
\DeclareMathSymbol{\Z}{\mathord}{AMSb}{"5A}
\DeclareMathSymbol{\N}{\mathord}{AMSb}{"4E}
\DeclareMathSymbol{\K}{\mathord}{AMSb}{"4B}
\DeclareMathSymbol{\M}{\mathord}{AMSb}{"4D}
\DeclareMathSymbol{\Q}{\mathord}{AMSb}{"51}
\DeclareMathSymbol{\Lset}{\mathord}{AMSb}{"4C}

\theorembodyfont{\rmfamily}
\theoremheaderfont{\bfseries\slshape}
\newtheorem{theorem}{Theorem}

\newtheorem{corollary}{Corollary}
\newtheorem{lem}{Lemma}

\newtheorem{definition}{Definition}
\newtheorem{ex}{Example}

\newtheorem{axiom}{Axiom}

\newcommand\defeq{\mathrel{:=}}

\newcommand\pink{\fcolorbox{Lavender}{Lavender}}
\newcommand\mycyan{\fcolorbox{SkyBlue}{SkyBlue}}
\newcommand\ie{{\em i.e.}}
\newcommand\eg{{\em e.g.}}
\newcommand\etal{{\em et al.}}
\definecolor{webmag}{rgb}{0.5,0,0.5}
\definecolor{myblue}{rgb}{0,0,1}
\newcommand{\nop}[1]{}

\title{On the Payoff Mechanisms in Peer-Assisted Services with Multiple
  Content Providers: \\ {\slshape Rationality and Fairness}}
\author{\authorblockN{Jeong-woo~Cho~and~Yung Yi,~\IEEEmembership{Members,~IEEE}}
\thanks{J. Cho is with the School of Information and Communication Technology at KTH Royal Institute of Technology, Sweden (email: \href{mailto:jwcho@kth.se}{jwcho@kth.se}).}
\thanks{Y. Yi is with Department of Electrical Engineering at KAIST (Korea Advanced Institute of Science and Technology), South Korea (email: \href{mailto:yiyung@kaist.edu}{yiyung@kaist.edu}).}
\thanks{A preliminary version of this
    work has been published in the proceedings of GameNets 2011.}
}

\makeatletter
\let\@copyrightspace\relax
\makeatother

\begin{document}
\maketitle

\newcommand{\expectation}{\textsf{E}}
\newcommand{\probability}{\textsf{P}}
\newcommand{\pdf}{\textsf{f}}
\newcommand{\slow}{\ell}
\newcommand{\nextline}{\mbox{}\\}
\newcommand{\ud}{\mathrm{d}}
\newcounter{tempcounter}
\newcounter{acounter}

\begin{abstract}
  This paper studies an incentive structure for cooperation and its
  stability in peer-assisted services when there exist multiple content
  providers, using a coalition game theoretic approach.  We first
  consider a generalized coalition structure consisting of multiple
  providers with many assisting peers, where peers assist providers to
  reduce the operational cost in content distribution. To distribute the
  profit from cost reduction to players (\ie, providers and peers), we
  then establish a generalized formula for individual payoffs when a
  ``Shapley-like'' payoff mechanism is adopted. We show that the grand
  coalition is {\em unstable}, even when the operational cost functions
  are concave, which is in sharp contrast to the recently studied case
  of a single provider where the grand coalition is stable.  We also
  show that irrespective of stability of the grand coalition, there
  always exist coalition structures which are not convergent to the
  grand coalition under a dynamic among coalition
  structures. Our results give us an incontestable fact that a provider
  does not tend to cooperate with other providers in peer-assisted
  services, and be separated from them. Three facets of the
  noncooperative (selfish) providers are illustrated; {\em (i)}
  underpaid peers, {\em (ii)} service monopoly, and {\em (iii)}
  oscillatory coalition structure. Lastly, we propose a stable payoff
  mechanism which improves fairness of profit-sharing by regulating the
  selfishness of the players as well as grants the content providers a
  limited right of realistic bargaining.  Our study opens many new
  questions such as realistic and efficient incentive structures and the
  tradeoffs between fairness and individual providers' competition in
  peer-assisted services.
\end{abstract}




\section{Introduction}\label{sec:intro}

\subsection{Motivation}

The Internet is becoming more content-oriented, and the need of
cost-effective and scalable distribution of contents has become the
central role of the Internet. Uncoordinated peer-to-peer (P2P) systems, \eg, BitTorrent,
have been successful in distributing contents, but the rights of the
content owners are not protected well, and most of the P2P contents are
in fact illegal. In its response, a new type of service, called {\em
  peer-assisted service,} has received significant attention these
days. In peer-assisted services, users commit a part of their resources
to assist content providers in content distribution with objective of
enjoying both scalability/efficiency in P2P systems and controllability in
client-server systems. Examples of application of peer-assisted services
include nano data center \cite{refValanciusGreen} and IPTV
\cite{refChaP2PTV}, where high potential of operational cost reduction
was observed. For instance, there are now 1.8 million IPTV subscribers
in South Korea, and the financial sectors forecast that by 2014 the IPTV
subscribers is expected to be 106 million \cite{refRNCOS}. However, it is clear that most users will not just
``donate'' their resources to content providers. Thus, the key factor to
the success of peer-assisted services is how to (economically)
incentivize users to commit their valuable resources and participate in
the service.

One of nice mathematical tools to study incentive-compatibility of
peer-assisted services is the coalition game theory which covers how
payoffs should be distributed and whether such a payoff scheme can be executed
by rational individuals or not. In peer-assisted services, the ``symbiosis''
between providers and peers are sustained when {\em (i)} the offered payoff
scheme guarantees fair assessment of players' contribution
under a provider-peer coalition and {\em (ii)} each individual has no
incentive to exit from the coalition. In the coalition game theory, the
notions of Shapley value and the core have been popularly applied to address
{\em (i)} and {\em (ii)}, respectively, when the entire players cooperate,
referred to as the {\it grand coalition}.  A recent paper by Misra \etal
\cite{refMisraP2P} demonstrates that the Shapley value approach is a
promising payoff mechanism to provide right incentives for cooperation
in a {\em single-provider} peer-assisted service.

\begin{figure}[t!]
  \centering
  \centerline{\includegraphics[width=1\columnwidth]{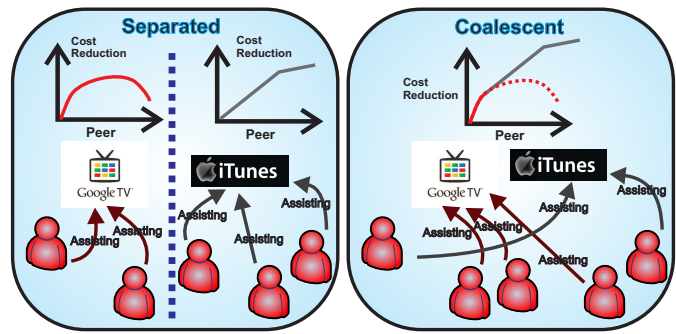}}
  \caption{Two coalition structures for a dual-provider peer-assisted
    service.} \label{fig:dualprovider}
\end{figure}

However, in practice, the Internet consists of multiple content
providers, even if only giant providers are counted. In the
multi-provider setting, users and providers are coupled in a more
complex manner, thus the model becomes much more challenging and even
the cooperative game theoretic framework itself is unclear, \eg, definition
of the worth of a coalition. Also, the results and their implications in the
multi-provider setting may experience drastic changes, compared to the
single-provider case.

The grand coalition is expected to be the ``best'' coalition in the
peer-assisted service with multiple providers in that it provides the
highest aggregate payoff. To illustrate, see an example in
Fig. \ref{fig:dualprovider} with two providers (Google TV and iTunes)
and a large number of peers. Consider two cooperation types: {\em (i)}
{\em separated,} where there exists a fixed partition of peers for each
provider, and {\em (ii)} {\em coalescent,} where each peer is possible
to assist any provider. In the separated case, a candidate payoff scheme
is based on the Shapley value in each disconnected coalition. In
the coalescent case, the Shapley value is also a candidate payoff scheme
after a worth function of the grand coalition is defined, where a reasonable
worth function\footnote{In Section~\ref{sec:worthpeer}, we establish
  that this definition is {\it derived} directly from an essential
  property of coalition.}~can be the total optimal profit, {\em maximized} over all
combinations of peer partitions to each provider.
Consequently, the total payoff for the coalescent case exceeds
that for the separated case, unless the two partitions of both cases are equivalent. Shapley value is defined by a few agreeable axioms, one of which is {\em efficiency}\footnote{To be discussed formally in Section \ref{sec:axiomatic}}, meaning that the {\em every} cent of coalition worth is distributed to players. Since smaller worth is shared out among players in the separated case, at least one individual is {\it underpaid} as compared with the coalescent case.  Thus, providers and users are
recommended to form the grand coalition and be paid off based on the
Shapley values.


However, it is still questionable whether peers are willing to stay in the grand
coalition and thus the consequent Shapley-value based payoff mechanism
is desirable in the multi-provider setting. In this paper, we anatomize
incentive structures in peer-assisted services with multiple content
providers and focus on stability issues from two different angles:
stability at equilibrium of Shapley value and convergence to the
equilibrium. We show that the Shapley payoff scheme may lead to unstable coalition structure, and propose a
different notion of payoff distribution scheme, $\chi$ value, under which
peers and providers stay in the stable coalition as well as better
fairness is guaranteed.

\subsection{Related Work}\label{sec:related}

The research on incentive structure in the P2P systems (\eg, BitTorrent)
has been studied extensively. To incapacitate free-riders in
 P2P systems, who only download contents but upload nothing, from
 behaving selfishly, a number of incentive mechanisms suitable for
 distribution of copy-free contents have been proposed (See
 \cite{refParkP2P} and references therein), using game theoretic
 approaches. Alternative approaches to exploit the potential of the P2P
systems for reducing the distribution (or operational) costs of the
copyrighted contents have been recently adopted by
\cite{refMisraP2P,refValanciusGreen}. To the best of our knowledge, the
work by Misra \etal \cite{refMisraP2P} is the first to study the
profit-sharing mechanism (payoff mechanism) of peer-assisted services.

Coalition game theory has been applied to model diverse networking
behaviors, where the main focus in most cases (\eg, \cite{refMisraP2P})
was to study the stability of a specific equilibrium \ie, the grand
coalition in connection with the notion of core.  Recently, Saad \etal
\cite{refSaadTutorial,refSaadHedonic}, discussed the stability and
dynamics of {\it endogenous} formation of general coalition
structures. In particular, \cite{refSaadHedonic} introduced a coalition
game model for self-organizing agents (\eg, unmanned aerial vehicles)
collecting data from arbitrarily located tasks in wireless networks and
proved the stability of the proposed algorithm by using {\it hedonic
  preference} (and dominance). In this paper, we use the stability
notion by Hart and Kurz \cite{refHartEndogenous} (see also
\cite{refTuticAD}) to study the dynamics of coalition structures in
peer-assisted services. The stability notion in \cite{refHartEndogenous}
is based on the preferences of any arbitrary {\it coalition} while the
hedonic coalition games are based on the preferences of {\it
  individuals}. Other subtle differences are described in
\cite{refCasajusStability}.

\subsection{Main Contributions and Organization}

We summarize our main contributions as follows:

\begin{compactenum}[1)]
\item Following the preliminaries in Section~\ref{sec:coalitiongame}, in
  Section~\ref{sec:issuemultiple}, we describe and propose the
  cooperative game theoretic framework of the peer-assisted service with
  multiple providers. After defining a worth function that is provably the
  unique feasible worth function satisfying two essential properties, \ie,
  feasibility and superadditivity of a coalition game, we provide a closed-form
  formula of the Shapley value for a general coalition with multiple
  providers and peers, where we take a fluid-limit approximation for
  mathematical tractability.  This is a non-trivial generalization of
  the Shapley value for the single-provider case in
  \cite{refMisraP2P}. In fact, our formula in Theorem
  \ref{th:advaluecoal} establishes the general Shapley value for
  distinguished {\it multiple} atomic players and infinitesimal players
  in the context of the Aumann-Shapley (A-S) prices \cite{refAS} in
  coalition game theory.

\item In Section~\ref{sec:gcinstability}, we discuss in various ways
  that the Shapley payoff regime cannot incentivize rational players to
  form the grand coalition, implying that {\it fair} profit-sharing and
  {\it opportunism} of players cannot stand together.  First, we prove
  that the Shapley value for the multiple-provider case is not in the
  core under mild conditions, \eg, each provider's cost function is
  concave. This is in stark contrast to the single-provider case where
  the concave cost function stabilizes the equilibrium.  Second, we
  study the dynamic formation of coalitions in peer-assisted services
  by introducing the notion of stability defined by the seminal work of
  Hart and Kurz \cite{refHartEndogenous}. Finally, we show that, if we
  adopt a Shapley-like payoff mechanism, called Aumann-Dr{\`eze} value,
  irrespective of stability of the grand coalition, there always exist
  initial states which do not converge to the grand coalition.

\item In Section~\ref{sec:critique}, we present three examples stating
  the problems of the non-cooperative peer-assisted service: {\em (i)}
  the peers are underpaid compared to their Shapley payoffs, {\em (ii)} a
  provider paying the highest dividend to peers monopolizes all
  peers, and {\em (iii)} Shapley value for each coalition gives rise to
  an oscillatory behavior of coalition structures. These examples
  suggest that the system with the separated providers may be even
  unstable as well as unfair in a peer-assisted service market.

\item In Section~\ref{sec:chivalue}, as a partial solution to the problems of Shapley-like payoffs (\ie, Shapley and Aumann-Dr{\`eze}), we propose an alternative payoff
  scheme, called $\chi$ value \cite{refChiValue}. This payoff mechanism
  is relatively {\it fair} in the sense that players, at the least,
  apportion the difference between the coalition worth and the sum of
  their fair shares, \ie, Shapley payoffs, and {\it stabilizes} the whole
  system. It is also practical in the sense that providers are granted a
  limited right of {\it bargaining}. That is, a provider may award an
  extra bonus to peers by cutting her dividend, competing with other
  providers in a fair way. More importantly, we show that authorities
  can effectively avoid unjust rivalries between providers by
  implementing a simplistic measure.

\end{compactenum}

After presenting a practical example of peer-assisted services with multiple providers in delay-tolerant networks in Section \ref{sec:app}, we conclude this paper.

\section{Preliminaries}\label{sec:coalitiongame}

Since this paper investigates a multi-provider case, where a peer can
choose any provider to assist, we start this section by defining a
coalition game with a peer partition (\ie, a coalition structure) and
introducing the payoff mechanism thereof.


\subsection{Game with Coalition Structure}\label{sec:pri}


A game with coalition structure is a triple $(N,v,\mathcal{P})$ where
$N$ is a player set and $v : 2^N \rightarrow \R $ ($2^N$ is the set of
all subsets of $N$) is a worth function, $v(\emptyset)=0$. $v(K)$ is
called the worth of a coalition $K\subseteq
N$. 
$\mathcal{P}$ is called a {\it coalition structure} for $(N,v)$; it is a
partition of $N$ where $C(i) \in \mathcal{P} $ denotes the
coalition containing player $i$. For your reference, a coalition structure $\mathcal{P}$ can be regarded as a set of disjoint coalitions. The {\em grand coalition} is the
partition $\mathcal{P} = \{ N \}$. For instance\footnotemark, a
partition of $N=\{1,2,3,4,5 \}$ is $\mathcal{P} = \{ \{1,2\}, \{ 3 ,4,5
\} \},$ $C(4) = \{3, 4,5\},$ and the grand coalition is
$\mathcal{P}=\{ \{1,2,3,4,5\} \}$. $\mathcal{P}(S)$ is the set of all
partitions of $S \subseteq N$. For notational simplicity, a game {\it without}
coalition structure $(N,v,\{ N \})$ is denoted by $(N,v)$.  A value of
player $i$ is an operator $\phi_i (N,v,\mathcal{P})$ that assigns a
payoff to player $i$. 
We define $\phi_K = \sum_{i\in K} \phi_i $ for all $K\subseteq N$.

\footnotetext{A player $i$ is an {\it element} of a coalition
  $C=C(i)$, which is in turn an {\it element} of a partition
  $\mathcal{P}$. $\mathcal{P}$ is an element of $\mathcal{P}(N)$
  {\it while} a subset of $ 2^N$.}


To conduct the equilibrium analysis of coalition games, the notion of {\it core} has been extensively used to study the stability of grand coalition $\mathcal{P} = \{ N \}$:
\begin{definition}[Core]\label{def:core}
  The core of a game $(N,v)$ is defined by:
  \begin{multline*}
    \bigg\{ \phi (N,v) ~\vert~
    \sum_{i\in N} \phi_i (N,v)= v(N) \cr \mbox{ and } \sum_{i\in K} \phi_i
    (N,v) \geq v(K), \forall K\subseteq N\bigg\}.
  \end{multline*}
\end{definition}
If a payoff vector $\phi(N,v)$ lies in the core, no player in $N$ has
an incentive to split off to form another coalition $K$ because the
worth of the coalition $K$, $v(K)$, is no more than the payoff sum $
\sum_{i\in K} \phi_i (N,v) $. Note that the definition of the core
hypothesizes that the grand coalition is already formed {\it
  ex-ante}. We can see the core as an analog of Nash equilibrium from
noncooperative games. Precisely speaking,
  it should be viewed as an analog of {\it strong Nash equilibrium}
  where no arbitrary coalition of players can create worth which is
  larger than what they receive in the grand coalition. If a payoff
vector $\phi (N,v)$ lies in the core, then the grand coalition is
stable with respect to any collusion to break the grand coalition.


\subsection{Shapley Value and Aumann-Dr{\`eze} Value}\label{sec:adintro}

On the premise that the player set is not partitioned,
\ie, $\mathcal{P}=\{N\}$, the Shapley value, denoted by $\varphi$ (not $\phi$), is popularly used as a fair
distribution of the grand coalition's worth to individual players,
defined by:
\begin{align}\label{eq:shapleyoriginal}
\textstyle \varphi_i(N,v) \! = \! \! \displaystyle \sum_{S \subseteq N
  \setminus \{ i \} } \textstyle \frac{|S|! (|N|-|S|-1)!}{|N|!} \left(
  v( S \cup \{i\} ) \! - \!  v(S) \right) .
 \end{align}
 Shapley \cite{refShapley} gives the following interpretation: ``{\em
   (i)} Starting with a single member, the coalition adds one player at
 a time until everybody has been admitted. {\em (ii)} The order in which
 players are to join is determined by chance, with all arrangements
 equally probable. {\em (iii)} Each player, on his admission, demands
 and is promised the amount which his adherence contributes to the value
 of the coalition.'' The Shapley value quantifies the above that is
 axiomatized (see Section~\ref{sec:axiomatic}) and has
 been treated as a worth distribution scheme. The beauty of the
 Shapley value lies in that the payoff ``summarizes'' in {\it one}
 number all the possibilities of each player's contribution in every
 coalition structure.

Given a coalition structure $\mathcal{P} \neq \{ N \}$, one can obtain the Aumann-Dr{\`eze} value (A-D value) \cite{refADValue} of player $i$, also denoted by $\varphi$, by taking $C(i)$, which is the coalition containing player $i$, to be the player set and by computing the Shapley value of player $i$ of the {\it reduced} game $(C(i),v)$. It is easy to see that the A-D value can be construed as a direct extension of the Shapley value to a
 game with coalition structure. Note that both Shapley value and A-D value are denoted by $\varphi$ because the only
 difference is the underlying coalition structure $\mathcal{P}$.

\subsection{Axiomatic Characterizations of Values}\label{sec:axiomatic}

We provide here an axiomatic characterization of the Shapley value \cite{refShapley}. 
\begin{axiom}[Coalition Efficiency, CE]\label{ax:ce}
$\sum_{j\in C} \phi_j (N,v,\mathcal{P})= v(C),~\forall C\in \mathcal{P}$.
\end{axiom}
\begin{axiom}[Coalition Restricted Symmetry, CS]\label{ax:cs}
If $j \in C(i)$ and $v(K\cup \{i \})= v(K\cup \{j \})$ for all $K \subseteq N \setminus \{i,j\}$, then $\phi_i(N,v,\mathcal{P})= \phi_j (N,v,\mathcal{P})$.
\end{axiom}
\begin{axiom}[Additivity, ADD]\label{ax:add}
For all coalition functions $v$, $v'$ and $i\in N$, $\phi_i (N, v + v', \mathcal{P}) = \phi_i (N, v, \mathcal{P}) + \phi_i (N,  v', \mathcal{P})$.
\end{axiom}
\begin{axiom}[Null Player, NP]\label{ax:np}
If $v(K \cup \{ i \}) = v(K)$ for all $K \subseteq N$, then $ \phi_i (N,v,\mathcal{P}) =0 $.
\end{axiom}

Recall that the basic premise of the Shapley value is that the player
set is not partitioned, \ie, $\mathcal{P} = \{N \}$. It is
well-known \cite{refChiValue,refShapley} that the Shapley value, defined
in \eqref{eq:shapleyoriginal}, is {\it uniquely} characterized by {\slshape
  CE}, {\slshape CS}, {\slshape ADD} and {\slshape NP} for $\mathcal{P} = \{N \}$. The A-D value is also {\it
  uniquely} characterized by {\slshape CE}, {\slshape CS}, {\slshape ADD} and {\slshape NP}
(Axioms \ref{ax:ce}-\ref{ax:np}), but in this case for arbitrary
coalition structure $\mathcal{P}$ \cite{refADValue}. In the literature,
\eg, \cite{refSaadTutorial,refPeleg}, the A-D value has been used to
analyze the {\it static} games where a coalition structure is {\it
  exogenously} given.

\begin{definition}[Coalition Independent, {\slshape CI}]\label{def:ci}
 If $i \in C \subseteq N$, $C \in  \mathcal{P}$ and $C \in  \mathcal{P}'$, then $\phi_i (N,v,\mathcal{P})= \phi_i (N,v,\mathcal{P}')$.
\end{definition}

From the definition of the A-D value, the payoff of player $i$ in
coalition $C(i)$ is affected neither by the player set $N$ nor by
coalitions $C \in \mathcal{P}$, $C \neq C(i)$. Note that only $C(i)$
contains the player $i$. Thus, it is easy to prove that the A-D value
is coalition independent.
From {\slshape CI} of the A-D value, in order to decide the payoffs of a game with general coalition structure $\mathcal{P}$, it suffices to decide the payoffs of players within each coalition, say $C\in \mathcal{P}$, without considering other coalitions $C \in \mathcal{P}$, $C \neq C(i)$. In other words, once we decide the payoffs of a coalition $C \in \mathcal{P}$, the payoffs remain unchanged even though other coalitions, $C' \in \mathcal{P}$, $C' \neq C$, vary. Thus, for any given coalition structure $\cal P$, any coalition $C \in \mathcal{P}$ is just two-fold in terms of the number of providers in $C$: {\em (i)} one
provider or {\em (ii)} two or more providers, as depicted in
Fig. \ref{fig:dualprovider}.




\begin{figure}[t!]
  \centering
  \centerline{\includegraphics[width=8cm]{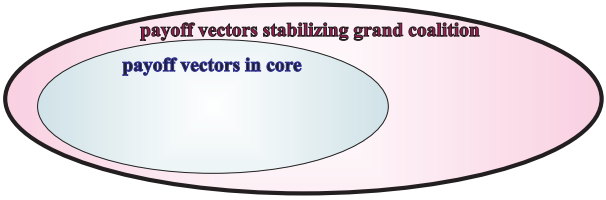}}
  \caption{If a payoff vector lies in the core, the grand coalition is
    stable \cite{refHartEndogenous}.} \label{fig:core}
\end{figure}
Yet another reason why {\slshape CI} attracts our attention is that it enables us to define the
stability of a game with coalition structure in the following simplistic way:
\begin{definition}[Stable Coalition Structure \cite{refHartEndogenous}]\label{def:stability}
We say that a coalition structure $\mathcal{P}'$ blocks $\mathcal{P}$, where $\mathcal{P}'$, $\mathcal{P} \in \mathcal{P}(N)$, with respect to $\phi$ if and only if there exists some $C \in \mathcal{P}'$ such that $\phi_i (N, v, \{ C, \cdots \}) > \phi_i (N,v,\mathcal{P})$ for all $i\in C$. In this case, we also say that $C$ blocks $\mathcal{P}$. If there does not exist any $\mathcal{P}'$ which blocks $\mathcal{P}$, $\mathcal{P}$ is called {\it stable}.
\end{definition}

Due to {\slshape CI} of the A-D value, all stability notions defined by the
seminal work of Hart and Kurz \cite{refHartEndogenous} coincide with the
above simplistic
definition, as discussed by Tutic \cite{refTuticAD}.
Definition~\ref{def:stability} can be intuitively interpreted that, if there exists any subset of players $C$ who improve their payoffs away from the current coalition structure, they {\it will} form a new coalition $C$. In other words, if a coalition structure $\mathcal{P}$ has any blocking coalition $C$, some rational players will break $\mathcal{P}$ to increase their payoffs. The basic premise here is that players are not clairvoyant, \ie, they are interested only in improving their instant payoffs in a myopic way. If a payoff vector lies in the core, the grand coalition
is stable in the sense of Definition~\ref{def:stability}, but the
converse is not necessarily true (see Fig. \ref{fig:core}).

\subsection{Comparison with Other Values}\label{sec:comparison}

In a particular category of games, called {\it
    voting} games or simple games, Banzhaf value as well as the Shapley
  value (also known as Shapley-Shubik index in this context) has been
  used in the literature (See, \eg, \cite{refCoreBring} and references
  therein). While the Shapley value has been extensively studied
  in many papers, there are no similar results for the Banzhaf
  value. For instance, the Shapley value is proven to lie in the core
  for a special type of games, called {\it convex} games, whereas there
  is no equivalent result for the Banzhaf value. Moreover, the Banzhaf
  value violates the efficiency axiom {\slshape CE} in Section
  \ref{sec:axiomatic} for a certain coalition structure $\mathcal{P} =
  \{N \}$, leading to inefficient sharing of the grand coalition worth.

As compared with Aumann-Dr{\`eze} value, a new
  value, referred to as {\it Owen} value (See, \eg, \cite[Chapter
  8.8]{refPeleg} or \cite[Chapter XII]{refOwen}) has emerged based on an
  {\it alternative} viewpoint on coalition, where a coalition forms not to share the coalition worth, but only to maximize their
  bargaining power with regard to division of the worth of the grand
  coalition. In other words, players form a labor union (coalition) to
  obtain a better bargaining position leading to a larger payoff,
  implying that the coalition efficiency axiom {\slshape CE} is also
  violated.
A delicate premise of this approach is that players must form the {\it grand coalition}, the worth of which is in fact the {\it largest} worth in superadditive games (See Definition \ref{def:superadditivity}), and bargain with each other at the same time. Also, in the context of P2P systems, whether it is more reasonable to
  nullify {\slshape CE} so that a portion of a worth of a coalition (peers and providers) $C \in \mathcal{P}$ becomes transferrable to other coalitions $C' \in \mathcal{P}$, $C \neq C'$, remains an open economic
  question.


\section{Coalition Game in Peer-Assisted Services}\label{sec:issuemultiple}


In this section, we first define a coalition game in a peer-assisted
service with multiple content providers by classifying the types of
coalition structures as {\it separated}, where a coalition includes
only one provider, and {\it coalescent}, where a coalition is allowed to
include more than one providers (see Fig. \ref{fig:dualprovider}).
To define
the coalition game, we will define a worth function of an arbitrary coalition
$S \subseteq N$ for such two cases.

\subsection{Worth Function in Peer-Assisted Services}\label{sec:worthpeer}

Assume that players $N$ are divided into two sets, the set of content
providers ${Z} \defeq \{ p_1, \cdots, p_{\zeta}\}$, and the set of peers
$ H \defeq \{n_{1} ,\cdots, n_\eta \} $, \ie, $N={Z} \cup H$. We also
assume that the peers are homogeneous, \eg, the same computing powers,
disk cache sizes, and upload bandwidths. Later, we discuss that our
results can be readily extended to nonhomogeneous peers. The set of
peers assisting providers is denoted by $ \bar H \subseteq H $ where $x \defeq |\bar H|/\eta$, \ie, the fraction of
assisting peers. We define the worth of a coalition $S$ to be the amount
of cost reduction due to cooperative distribution of the contents by the
players in $S$ in both separated and coalescent cases.


\smallskip
\noindent\underline{\bf Separated case}: Denote by $\Omega_p^{\eta}(x(S))$ the
operational cost of a provider $p$ when the coalition $S$ consists of
a single provider $p$ and $x(S) \cdot \eta$ assisting peers. Since the
operational cost cannot be negative and may decrease with the number of assisting peers, we assume the following to simplify the
 exposition:
 \begin{compactitem}
 \item Assumption: $\Omega_p^{\eta}(x)$ is non-increasing in $x$ for all $p \in {Z}$.
\end{compactitem}
Note that from the
homogeneity assumption of peers, the cost function depends only on the
fraction of assisting peers. Then, we define the worth function
$\hat{v}(S)$ for a  coalition $S$ having a single provider as:
 \begin{align}\label{eq:worthfirst}
   \hat v (S ) \defeq \Omega_p^{\eta}(0)- \Omega_p^{\eta}(x(S))
\end{align}
where $\Omega_p^{\eta}(0)$ corresponds to the cost when there are no
assisting peers. For a coalition $S$ with no provider,
  we simply have $\hat{v}(S) \defeq 0.$ For notational
simplicity, $x(S)$ is henceforth denoted by $x,$ unless confusion
arises.

\smallskip
\noindent\underline{\bf Coalescent case}:
In contrast to the separated case, where a coalition includes a single
provider, the worth for the coalescent case is not clear yet, since
depending on which peers assist which providers the amount of cost reduction
may differ. One of reasonable definitions would be the maximum worth out of
all peer partitions, \ie, the worth for the coalescent case is defined
by: for a coalition with at least two providers,
\begin{align}\label{eq:worthmultiple}
\textstyle v(S) \!\defeq\! \max \!\left\{\! \displaystyle{\sum_{C\in \mathcal{P}}} \hat v(C)
  \Big\vert {\mathcal{P} \!\in\! \mathcal{P}(S)}\!\!~\mbox{ s.t. }|Z\! \cap\! C| \!=\! 1,~ \forall C\!\in\! \mathcal{P} \! \right\},
\end{align}
and $v(S) \defeq \hat{v}(S)$ for a coalition $S$ with at most one provider.
The definition above implies that we {\it view} a coalition containing
more than one provider as the most productive coalition whose worth is
{\it maximized} by choosing the optimal partition $\mathcal{P}^* $ among
all possible partitions of $S$.  Note that \eqref{eq:worthmultiple} is
consistent with the definition \eqref{eq:worthfirst} for $|{Z}\cap S|
= 1$, \ie, $v(S)= \hat v(S)$ for $|{Z} \cap S| = 1$.

 Five remarks are in order. First, as opposed to \cite{refMisraP2P} where $ \hat v(\{ p \})= \eta R - \Omega_p^{\eta}(0)$ ($R$ is the subscription fee paid by any peer), we simply assume that $ \hat v(\{ p \} )= 0$. Note that, as discussed in \cite[Chapter 2.2.1]{refPeleg}, it is no loss of generality to assume that, initially, each provider has earned no money. In our context, this means that it does not matter how much fraction of peers is subscribing to each provider because each peer has already paid the subscription fee to providers {\it ex-ante}.

 Second, $\Omega_p^{\eta}(x)$ may not be decreasing because, for
 example, electricity expense of the computers and the maintenance cost
 of the hard disks of peers may exceed the cost reduction due to peers'
 assistance in content distribution, \eg, Annualized Failure Rate (AFR)
 of hard disk drives is over 8.6\% for three-year old ones \cite{refHardAFR}.

Third, the worth function in peer-assisted services can reflect the
diversity of peers. It is not difficult to extend our result to the case
where peers belong to distinct classes. For example, peers may be
distinguished by different upload bandwidths and different hard disk
cache sizes. A point at issue for the multiple provider case is whether
peers who are {\it not} subscribing to the content of a provider may be
allowed to assist the provider or not. On the assumption that the
content is ciphered and not decipherable by the peers who do not know
its password which is given only to the subscribers, providers will
allow those peers to assist the content distribution. Otherwise, we can
easily reflect this issue by dividing the peers into a number of classes
where each class is a set of peers subscribing to a certain
content. 

Fourth, it should be clearly understood that our worth function \eqref{eq:worthmultiple} does not encompass more than just the peer-partition optimization. That is, we speculate that cooperation among providers might lead to further expenses cut by optimizing their network resources. We recognize the lack of this `added bonus' to be the major weakness in our model.

Lastly, it should be noted that the worth function in \eqref{eq:worthmultiple} is selected in order to satisfy two properties. First of all, it follows from the definition of $v$ in
  \eqref{eq:worthmultiple} that no other coalition function $v'(\cdot)$ can be greater than
  $v(\cdot)$, \ie, $v(\cdot) \geq v'(\cdot)$ because $v$ is the total
  cost reduction that is {\it maximized} over all possible peer partitions to
  each provider.
\begin{definition}[Feasibility]\label{def:feasibility}
For all worth function $v'(\cdot)$, we have $v(S)\geq  v'(S)$ for all $S \subseteq N$.
\end{definition}
The second property, superadditivity, is one of the most elementary properties, which ensures that the core is nonempty by appealing to Bondareva-Shapley Theorem \cite[Theorem 3.1.4]{refPeleg}.
\begin{definition}[Superadditivity]\label{def:superadditivity}
A worth $v$ is superadditive if $\left( S,T\subseteq N \mbox{ and }  S \cap T =\emptyset \right) \Rightarrow v(S\cup T) \geq v(S)+v(T) $.
\end{definition}

The following lemma holds by the fact that a feasible worth function cannot be greater than \eqref{eq:worthmultiple}, \ie, the largest worth.

\begin{lem}\label{lem:sup}
When the worth for the separated case is given by
  \eqref{eq:worthfirst}, for the coalescent case, there exists a unique worth function that is {\it both} superadditive and feasible, given
  by \eqref{eq:worthmultiple}.
\end{lem}
\begin{proof}
Suppose we have a superadditive worth $v'$. Firstly, it follows directly from the assumption (the worth function for the separate case is \eqref{eq:worthfirst}) that  $v'(S) = \hat v(S)$ if $S$ includes one provider. {\em (i)
    Feasibility}: It follows from the definition of feasibility that we have $v(\cdot) \geq v'(\cdot)$ because $v(S)$ is the maximum over all possible partitions $\mathcal{P} \in \mathcal{P}(S)$. {\em (ii) Superadditivity}: In the meantime, since $v'$ is
  superadditive, it must satisfy $v'(S\cup T) \geq v'(S)+v'(T)$ for all
  disjoint $S$, $T \subseteq N$. This in turn implies $v'(S) \geq
  \sum_{C\in \mathcal{P}} v'(C)$ for {\bfseries\slshape all} $ \mathcal{P}$ such that $\mathcal{P} \in
  \mathcal{P}(S)$. The right-hand side $\sum_{C\in \mathcal{P}} v'(C)$ should coincide with
  $v(S)$ for {\bfseries\slshape some} $\mathcal{P}=\mathcal{P}^*$ such that $|{Z} \cap
  C| = 1$ for all $C\in \mathcal{P}^*$ (See \eqref{eq:worthmultiple}), where $\mathcal{P}^*$ is the peer partition which maximizes $v(S)$. Therefore, we have $v'(S) \geq v(S)$. Combining this with $v(\cdot) \geq v'(\cdot)$ uniquely determines $v'(\cdot)=v(\cdot)$.
\end{proof}

 In light of this lemma, we can restate that our objective in this paper is to analyze the incentive structure of peer-assisted services when the worth of coalition is feasible and superadditive. This objective in turn implies the form of worth function in \eqref{eq:worthmultiple}.

 \setcounter{tempcounter}{\value{equation}}
\setcounter{equation}{0}
\renewcommand{\theequation}{FluidAD\arabic{equation}}
\begin{figure*}[t]
\begin{minipage}{\textwidth}
\begin{equation}\label{eq:admultivalue}\left\{ \begin{array}{ll}
\widetilde{\varphi}_{p}^{{\bar Z}} (x) = \widetilde\Omega_p (0) - \sum_{S \subseteq {\bar Z} \setminus \{p\}} \int_{0}^{1} u^{|S|} (1-u)^{|{\bar Z}|-1-|S|} \left( M_{\Omega}^{S \cup \{p\}} (u x)- M_{\Omega}^{S} (u x) \right)  \ud u , & \mbox{for } p \in {\bar Z} \\
\widetilde{\varphi}_{n}^{{\bar Z}} (x) = - \sum_{S \subseteq {\bar Z}} \int_{0}^{1} u^{|S|} (1-u)^{|{\bar Z}|-|S|} \frac{\ud M_{\Omega}^S }{\ud x } (u x ) \ud u ,  & \mbox{for } n \in \bar H .
\end{array} \right.
\end{equation}
\begin{equation}\label{eq:addualvalue}\left\{ \begin{array}{ll}
\widetilde{\varphi}_{p}^{\{ p,q \}} (x) = {\widetilde\Omega}_p (0) - \int_{0}^{1}  u M_{\Omega}^{\{p,q\}} (u x ) \ud u  - \int_{0}^{1} (1-u)   M_{\Omega}^{\{p\}} (u x)   \ud u + \int_{0}^{1} u   M_{\Omega}^{\{q\}} (u x)  \ud u , & \mbox{(} p,~q \mbox{ are interchangeable)}\\
\widetilde{\varphi}_{n}^{\{ p,q \}} (x) = - \int_{0}^{1} u^2 \frac{\ud M_{\Omega}^{ \{ p,q \} } }{\ud x } (u x ) \ud u - \sum_{i\in \{ p,q \} } \int_{0}^{1} u(1-u)  \frac{\ud M_{\Omega}^{\{i\}}}{\ud x} (u x) \ud u , &\mbox{for } n \in \bar H  .
\end{array} \right.
\end{equation}
\end{minipage}
\put(-5,-41){\line(-1,0){508}}
\vspace{1mm}
\caption{Fluid Aumann-Dr{\`eze} payoff formula for multi-provider coalitions, construed as an extension of Aumann-Shapley prices to multiple atomic players.}\label{fig:equation}
\end{figure*}
\renewcommand{\theequation}{\arabic{equation}}
\setcounter{equation}{\value{tempcounter}}

\subsection{Fluid Aumann-Dr{\`eze} Value for Multi-Provider Coalitions}\label{sec:insensitivead}


So far we have defined the worth of coalitions. Now let us {\it
  distribute} the worth to the players for a given coalition structure
$\mathcal{P}$. Recall that the payoffs of players in a coalition are
independent from other coalitions by the definition of A-D
payoff. Pick a coalition $C$ without loss of generality, and denote
the set of providers in $C$ by $\bar Z \subseteq Z$. With slight notational
abuse, the set of peers assisting $\bar Z$ is denoted by $\bar H$. Once
we find the A-D payoff for a coalition consisting of arbitrary provider
set ${\bar Z} \subseteq Z $ and assisting peer set $\bar H \subseteq H$, the payoffs
for the separated and coalescent cases in Fig. \ref{fig:dualprovider}
follow from the substitutions $\bar Z = Z$ and $\bar Z = \{ p \}$,
respectively. In light of our discussion in Section \ref{sec:adintro},
it is more reasonable to call a Shapley-like payoff mechanism `A-D payoff' and
`Shapley payoff' respectively for the partitioned and non-partitioned
games $(N,v, \{{\bar Z} \cup \bar H , \cdots \})$ and $(N,v, \{ Z \cup H
\}
)$\footnotemark. 

 \footnotetext{On the contrary, the term `Shapley payoff' was used in \cite{refMisraP2P} to refer to the payoff for the game $(N,v, \{{\bar Z} \cup \bar H , \cdots \})$ where a proper subset of the peer set assists the content distribution.}

\smallskip
\noindent \underline{\bf Fluid Limit}: We adopt the limit axioms for a large
 population of users to overcome the computational hardness of the A-D
 payoffs:
 \begin{align}\label{eq:normalizecost}
\textstyle  \lim_{\eta \to \infty} \widetilde{\Omega}_{p}^{\eta} (\cdot) = \widetilde{\Omega}_{p} (\cdot ) ~~\mbox{   where  } \widetilde{\Omega}_{p}^{\eta} (\cdot)= \frac{1}{\eta} {\Omega}_{p}^{\eta} (\cdot)
  \end{align}
which is the asymptotic operational cost per peer in the system with a
large number of peers. We drop superscript $\eta$ from notations to denote their limits as $\eta \to \infty$. From the assumption $\Omega_p^{\eta}(x)>0$, we have $\widetilde\Omega_p(x) \geq 0$. To avoid trivial cases, we also assume ${\widetilde\Omega}_p (x) $ is not constant in the interval $x \in [0, 1]$ for any $p \in { Z}$.
We also introduce the payoff of each provider per user, defined as $\widetilde{\varphi}^{\eta}_{p}  \defeq \frac{1}{\eta} \varphi^{\eta}_{p} $.
We now derive the fluid limit equations of the payoffs, shown in Fig. \ref{fig:equation}, which can be obtained as $\eta \to \infty$. The proof of the following theorem is given in Appendix \ref{proof:th:advaluecoal}.
\begin{theorem}[A-D Payoff for Multiple Providers]\label{th:advaluecoal}
  As $\eta \to \infty$, the A-D payoffs of providers and peers
  under an arbitrary coalition $C = \bar{Z} \cup \bar{H}$ converge to \eqref{eq:admultivalue} in Fig. \ref{fig:equation} where $ M_{\Omega}^{S} (x) \defeq \min \left\{  \sum_{i\in S} \widetilde\Omega_i (y_i) ~\big\vert~ \sum_{i\in S} y_i \leq x ,~ y_i \geq 0 \right\}$ and $M_{\Omega}^{\emptyset} (x) \defeq 0$. Note that $M_{\Omega}^{\{p\}}(x)= {\widetilde\Omega}_p (x)$. 
\end{theorem}
The following corollaries are immediate as special cases of
  Theorem~\ref{th:advaluecoal}, which we will use in Section~\ref{sec:critique}.
\begin{corollary}[A-D Payoff for Single Provider]\label{cor:advalue}
As $\eta \to \infty$, the A-D payoffs of providers and peers who belong to a single-provider coalition, \ie, ${\bar Z} = \{ p \}$, converge to:
\begin{equation}\label{eq:adsinglevalue}\left\{ \begin{array}{ll}
\widetilde{\varphi}_{p}^{\{p\}} (x) = {\widetilde\Omega}_p (0) - \int_{0}^{1} M_{\Omega}^{\{p\}} (u x) \ud u , & \\
\widetilde{\varphi}_{n}^{\{p\}} (x) = - \int_{0}^{1} u\frac{\ud M_{\Omega}^{\{p\}} }{\ud x}  ( u x) \ud u , &\mbox{for } n \in \bar H  .
\end{array} \right.
\end{equation}
\end{corollary}
\begin{corollary}[A-D Payoff for Dual Providers]\label{cor:addualvalue}
As $\eta \to \infty$, the A-D payoffs of providers and peers who belong to a dual-provider coalition, \ie, ${\bar Z} = \{ p,q \}$, converge to \eqref{eq:addualvalue}.
\end{corollary}

  Note that our A-D payoff formula in Theorem~\ref{th:advaluecoal}
  generalizes the formula in Misra \etal \cite[Theorem 4.3]{refMisraP2P}
  (\ie, $|{Z}|=1$). It also establishes the A-D values for distinguished
  {\it multiple} atomic players (the providers) and infinitesimal
  players (the peers), in the context of the Aumann-Shapley (A-S) prices
  \cite{refAS} in coalition game theory.


  Our formula for the peers is interpreted as follows: Take the
  second line of \eqref{eq:addualvalue} as an example. Recall the
  definition of the Shapley value \eqref{eq:shapleyoriginal}. The payoff
  of peer $n$ is the {\it marginal} cost reduction $v(S \cup \{n\})-
  v(S)$ that is {\it averaged} over all equally probable arrangements,
  \ie, the orders of players. It is also implied by
  \eqref{eq:shapleyoriginal} that the {\it expectation} of the marginal
  cost is computed under the assumption that the events $|S|=y$ and
  $|S|=y'$ for $y \neq y'$ are {\it equally probable}, \ie,
  $\probability ( |S|=y) = \probability (|S|=y')$. Therefore, in our
  context of infinite player game in Theorem \ref{th:advaluecoal}, for
  every values of $u x$ along the interval $[0, x]$, the subset $S
  \subseteq {\bar Z} \cup \bar H$ contains $u x$ fraction of the
  peers. More importantly, the probability that each provider is a
  member of $S$ is simply $u$ because the size of peers in $S$, $\eta u
  x$, is infinite as $\eta \to \infty$ so that the {\it size} of $S$ is
  not affected by whether a provider belongs to $S$ or not. Therefore,
  the marginal cost reduction of each peer on the condition that both
  providers are contained in $S$ becomes $ - u^2 \frac{\ud M_{\Omega}^{
      \{ p,q \} } }{\ud x } (u x ) $. Likewise, the marginal cost
  reduction of each peer on the condition that only one provider is in
  the coalition is $ - u(1-u) \frac{\ud M_{\Omega}^{\{p\}}}{\ud x } (u
  x)$.


\section{Instability of the Grand Coalition}\label{sec:gcinstability}

In this section, we study the stability of the grand coalition to see if {\it rational} players are willing to form the grand coalition, only under which they can be paid their respective {\it fair} Shapley payoffs. The key message of this section is that
the rational behavior of the providers makes the Shapley
value approach {\it unworkable} because the major premise of the Shapley
value, the grand coalition, is not formed in the multi-provider games.

\begin{figure*}[t!]
\begin{center}
  \subfigure{\label{fig:ex1a}
        \includegraphics[width=4.25cm]{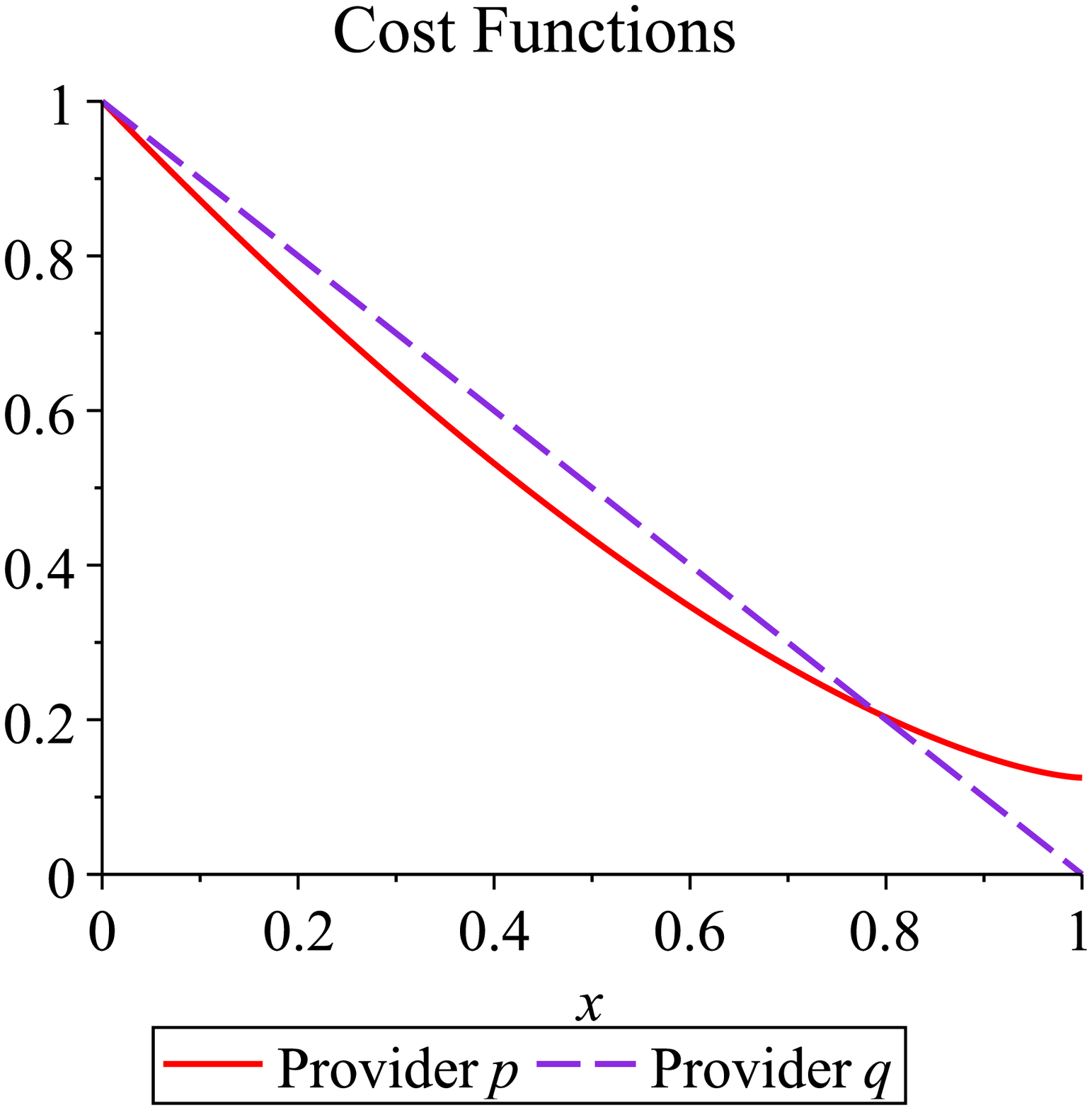}}
 \subfigure{\label{fig:ex1b}
        \includegraphics[width=4.25cm]{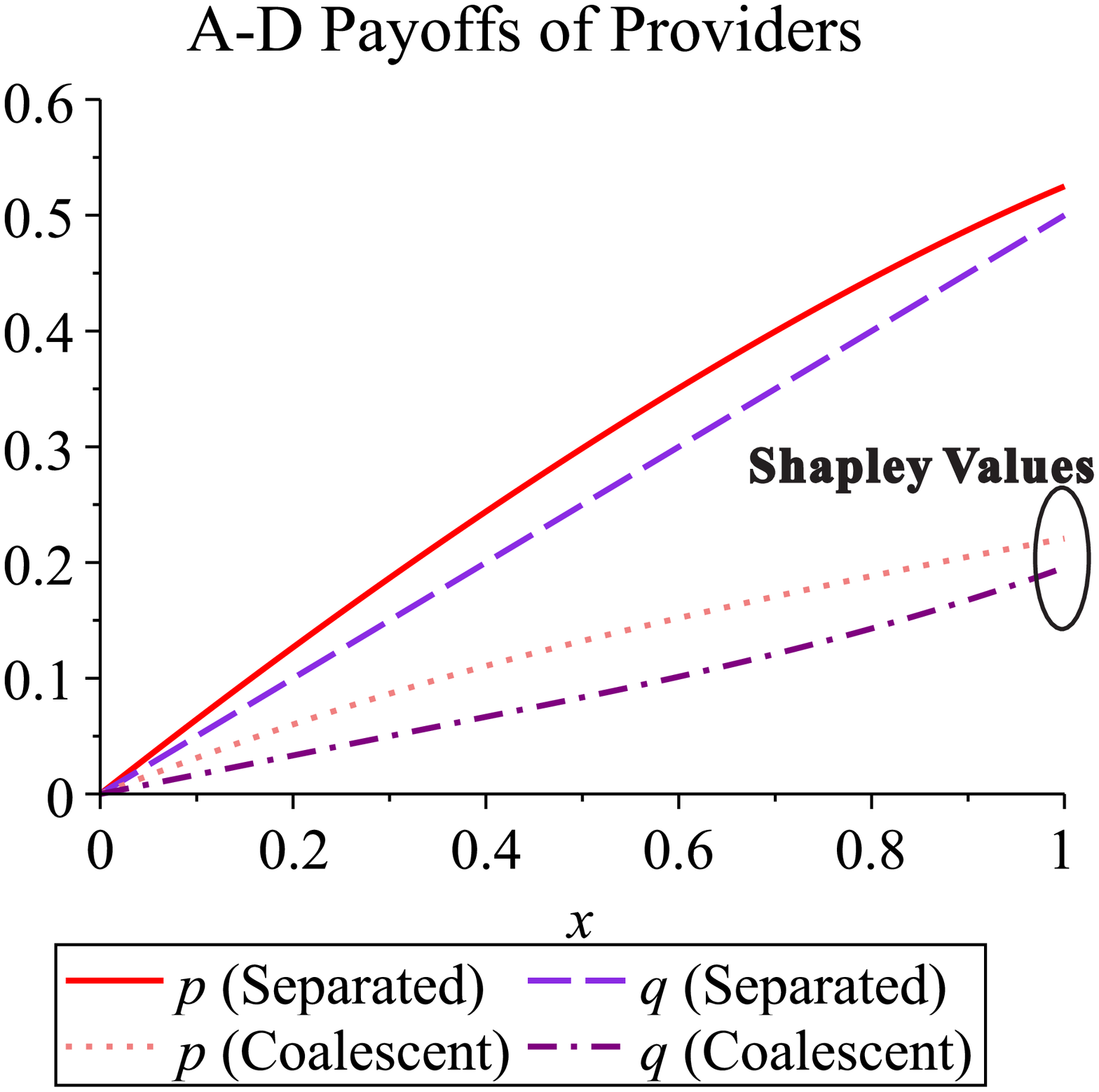}}
  \subfigure{\label{fig:ex1c}
        \includegraphics[width=4.25cm]{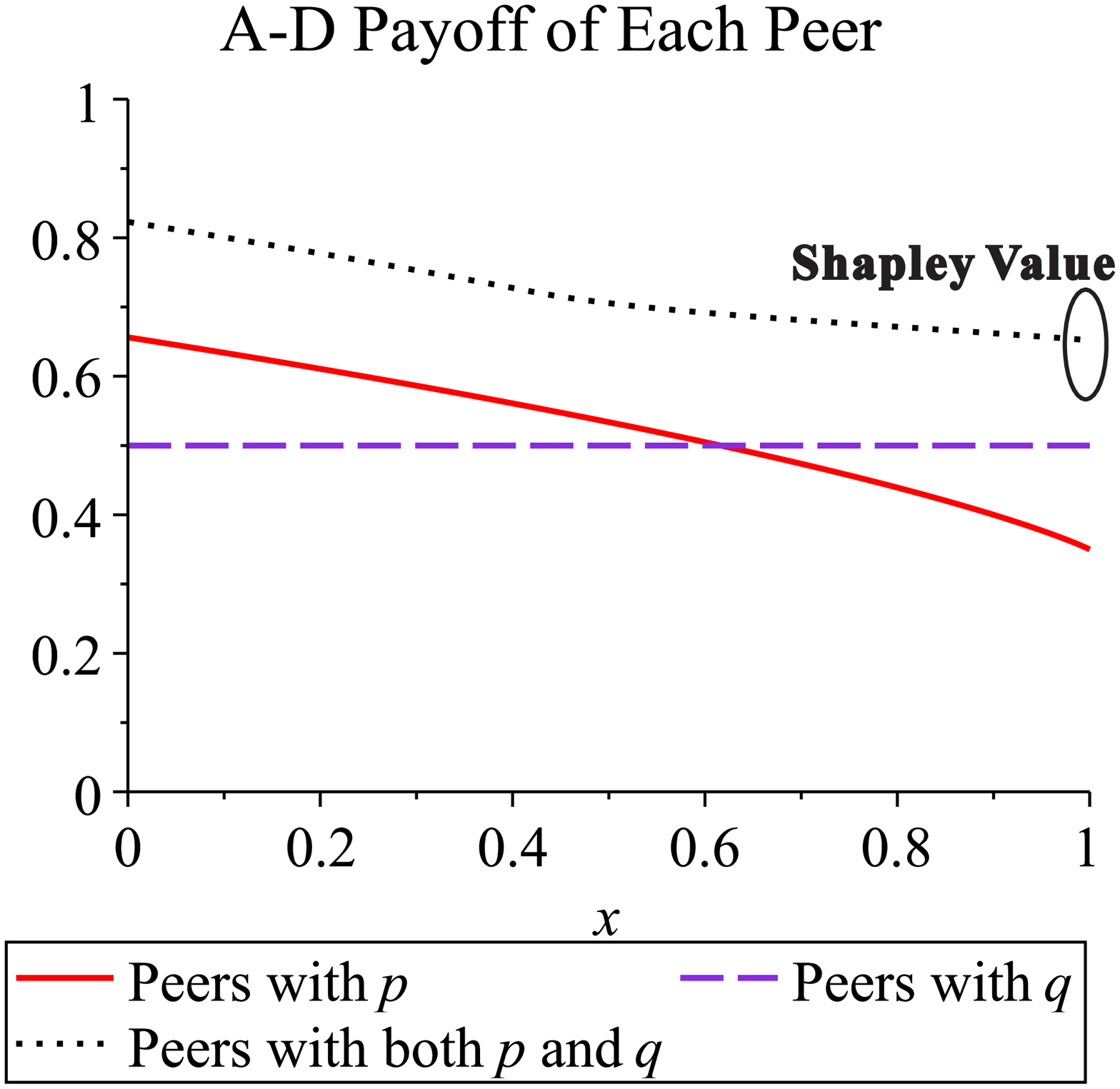}}
\caption{Example \ref{ex:unfairness}: A-D Payoffs of Two Providers and Peers for Convex Cost Functions.} \label{fig:ex1}
\end{center}
\end{figure*}

\subsection{Stability of the Grand Coalition}\label{sec:core}

Guaranteeing the stability of a payoff vector has been an important
topic in coalition game theory.  For the single-provider case,
$|{Z}|=1$, it was shown in \cite[Theorem 4.2]{refMisraP2P} that, if the
cost function is decreasing and concave, the Shapley incentive structure
lies in the core of the game. What if for $|{Z}|\geq 2$? Is the grand
coalition stable for the multi-provider case? Prior to addressing this
question, we first define the following:
\begin{definition}[Noncontributing Provider]\label{def:noncontributing}
A provider $p \in { Z}$ is called {\it noncontributing} if $M_{\Omega}^{ Z}(1) - M_{\Omega}^{{ Z}\setminus \{ p\}}(1) = {\widetilde\Omega}_p(0)$.
\end{definition}
To understand this better, note that the above expression is equivalent to the following:
\begin{align}\label{eq:noncontdef}
  \textstyle  \displaystyle\sum_{i\in { Z}}\textstyle {\widetilde\Omega}_i(0) -
    M_{\Omega}^{ Z}(1)  =  \displaystyle\sum_{i\in { Z}\setminus \{ p \}
    }\textstyle {\widetilde\Omega}_i(0) - M_{\Omega}^{{ Z}\setminus \{ p\}}(1)
\end{align}
which implies that there is no difference in the total cost reduction,
irrespective of whether the provider $p$ is in the provider set or not.
Interestingly, if all cost functions are concave, there exists at least
one noncontributing provider.
\begin{lem}\label{lem:concave}
Suppose $|{ Z}|\geq 2$. If ${\widetilde\Omega}_p(\cdot)$ is concave for all $p\in { Z}$, there exist $|Z|-1$  noncontributing providers.
\end{lem}
To prove this, recall the definition of $M_{\Omega}^{ Z}(\cdot)$:
$$ \textstyle M_{\Omega}^{{ Z}} (x) = \min_{y\in Y(x)}   \sum_{i\in { Z}} {\widetilde\Omega}_i (y_i) $$ $$ \textstyle \mbox{where}~ Y(x)\defeq \{ (y_1,\cdots,y_{|{ Z}|}) ~\big\vert~  \sum_{i\in { Z}} y_i \leq x ,~ y_i \geq 0 \}.$$
  Since the summation of concave functions is concave and the minimum of a concave function  over a convex feasible region $Y(x)$ is an {\it extreme} point of $Y(x)$ as shown in \cite[Theorem 3.4.7]{refNonlinear}, we can see that the solutions of the above minimization are the extreme points of $\{ (y_1, \cdots,y_{|{ Z}|}) ~\vert~ \sum_{i\in { Z}} y_i \leq x ,~ y_i \geq 0\}$, which in turn imply $y_i=0$ for $|Z|-1$ providers in $Z$. Note that the condition $|Z|\geq 2$ is {\it necessary} here.

 We are ready to state the following theorem, a direct consequence of Theorem \ref{th:advaluecoal}. Its proof is in Appendix \ref{proof:th:adnotcore}.
\begin{theorem}[Shapley Payoff Not in the Core]\label{th:adnotcore}
If there exists a noncontributing provider, the Shapley payoff for the game $({ Z}\cup  H, v)$ does not lie in the core.
\end{theorem}

It follows from Lemma \ref{lem:concave} that, if all operational cost functions are concave and $|{ Z}|\geq 2$, the Shapley payoff does not lie in the core.
This result appears to be in good agreement with our usual intuition. If
there is a provider who does not contribute to the coalition at all in
the sense of \eqref{eq:noncontdef} and is still being paid due to her
potential for imaginary contribution assessed by the Shapley formula
\eqref{eq:shapleyoriginal}, which is not actually exploited in the
current coalition, other players may improve their payoff sum by
expelling the noncontributing provider.

The condition $|{ Z}|\geq 2$ plays an essential role in the theorem. For
$|{ Z}|\geq 2$, the concavity of the cost functions leads to the Shapley
value not lying in the core, whereas, for the case $|{ Z}|=1$, the
concavity of the cost function is proven to make the Shapley incentive
structure lie in the core \cite[Theorem 4.2]{refMisraP2P}.

\subsection{Convergence to the Grand Coalition}\label{sec:convergence}

The notion of the core lends itself to the stability analysis of the
grand coalition {\it on the assumption} that the players are already in
the equilibrium, \ie, the grand coalition. However,
Theorem~\ref{th:adnotcore} still leaves further questions unanswered.
In particular, for the non-concave cost functions, it is unclear if the
Shapley value is not in the core, which is still an open problem. We
rather argue here that, whether the Shapley value lies in the core or
not, the grand coalition is unlikely to occur by showing that the grand
coalition is not a global attractor under some conditions.

To study the convergence of a game with coalition structure to the grand coalition, let us recall Definition \ref{def:stability}.
It is interesting that, though the notion of stability was not used in \cite{refMisraP2P}, one main argument of this work was that the system with one provider would converge to a full sharing mode, \ie, the grand coalition, hinting the importance of the following convergence result with multiple providers. The proof of the following theorem is given in Appendix \ref{proof:th:convergencetocore}.

\begin{theorem}[A-D Payoff Doesn't Lead to Grand Coalition]\label{th:convergencetocore}
Suppose $|{ Z}|\geq 2$ and ${\widetilde\Omega}_p (y) $ is not constant in the interval $y \in [0, x]$ for any $p \in { Z}$ where $x=|\bar H|/| H|$. The following holds for all $p\in { Z} $ and $n \in \bar H$.
 \begin{compactitem}
 \item The A-D payoff of provider $p$ in coalition $\{p \} \cup \bar H$ is larger than that in all coalition $T \cup \bar H$ for $\{p\} \subsetneq T \subseteq { Z}$.
 \item The A-D payoff of peer $n$ in coalition $\{p\} \cup \bar H$ is smaller than that in all coalition $T\cup \bar H$ for $\{p\} \subsetneq T \subseteq { Z}$.
 \end{compactitem}
\end{theorem}

In plain words, a provider, who is in cooperation with a peer set, will receive the highest dividend when she cooperates only with the peers excluding other providers whereas each peer wants to cooperate with as many as possible providers. It is surprising that, for the multiple provider case, \ie, $|{ Z}| \geq 2$, each provider benefits from forming a single-provider coalition {\it whether} the cost function is concave {\it or not}. There is no {\it positive} incentives for providers to cooperate with each other under the implementation of A-D payoffs. On the contrary, a peer always looses by leaving the grand coalition.



Upon the condition that each provider begins with a single-provider
coalition with a sufficiently large number of peers, one cannot reach the grand
coalition because some single-provider coalitions are already {\it
  stable} in the sense of the stability in Definition
\ref{def:stability}. That is, the grand coalition is not the global
attractor. For instance, take $\mathcal{P} = \{ \{p \} \cup H , \cdots \}$ as the current coalition structure where all peers are possessed by provider $p$. Then it follows from Theorem \ref{th:convergencetocore} that players cannot make any transition from $\mathcal{P}$ to $\{ \Phi \cup H , \cdots \}$ where $\Phi \subseteq Z$ is any superset of $\{p\}$ because provider $p$ will not agree to do so.

\begin{figure*}[t!]
\begin{center}
  \subfigure{\label{fig:ex2a}
        \includegraphics[width=4.25cm]{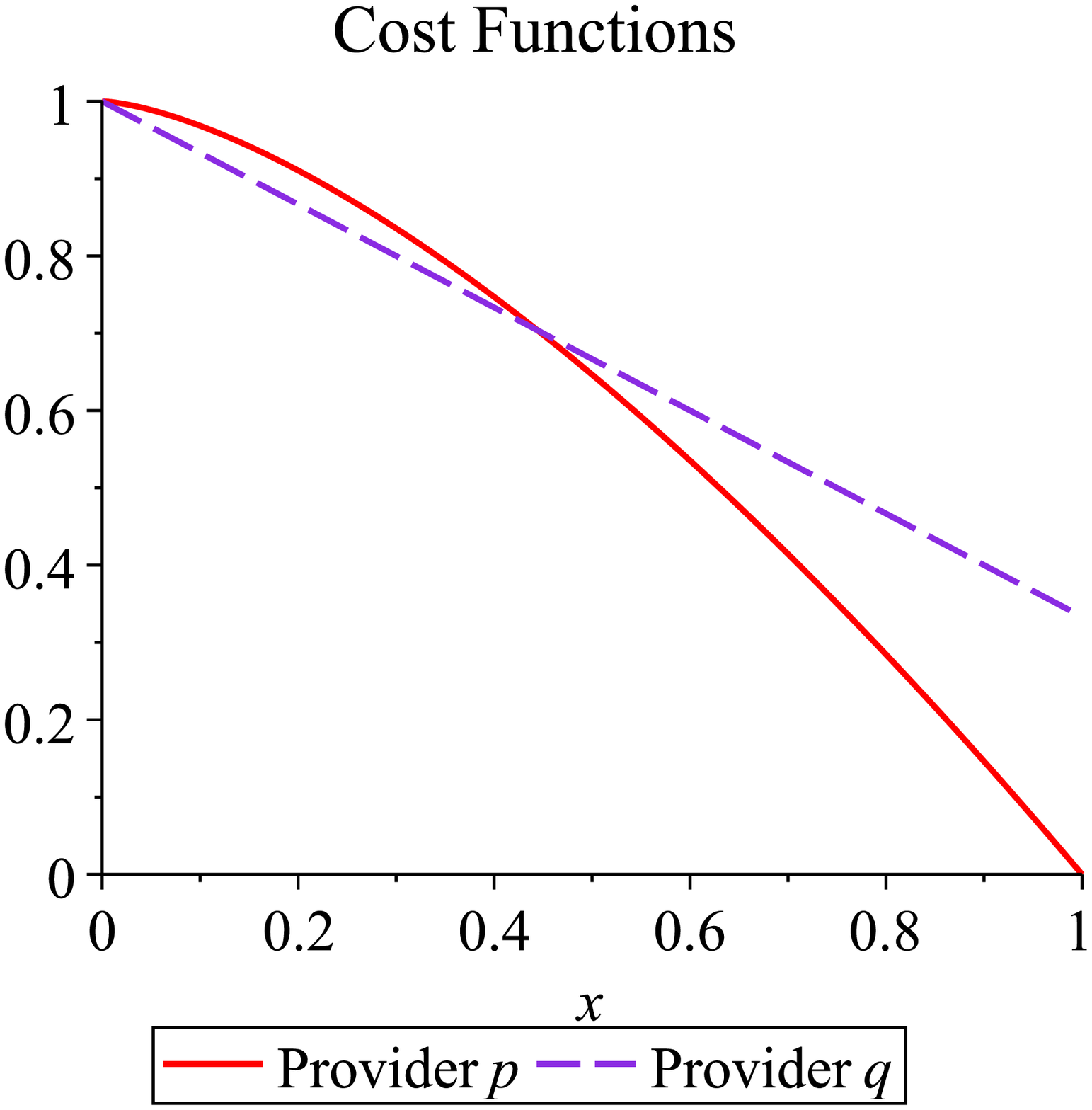}}
 \subfigure{\label{fig:ex2b}
        \includegraphics[width=4.25cm]{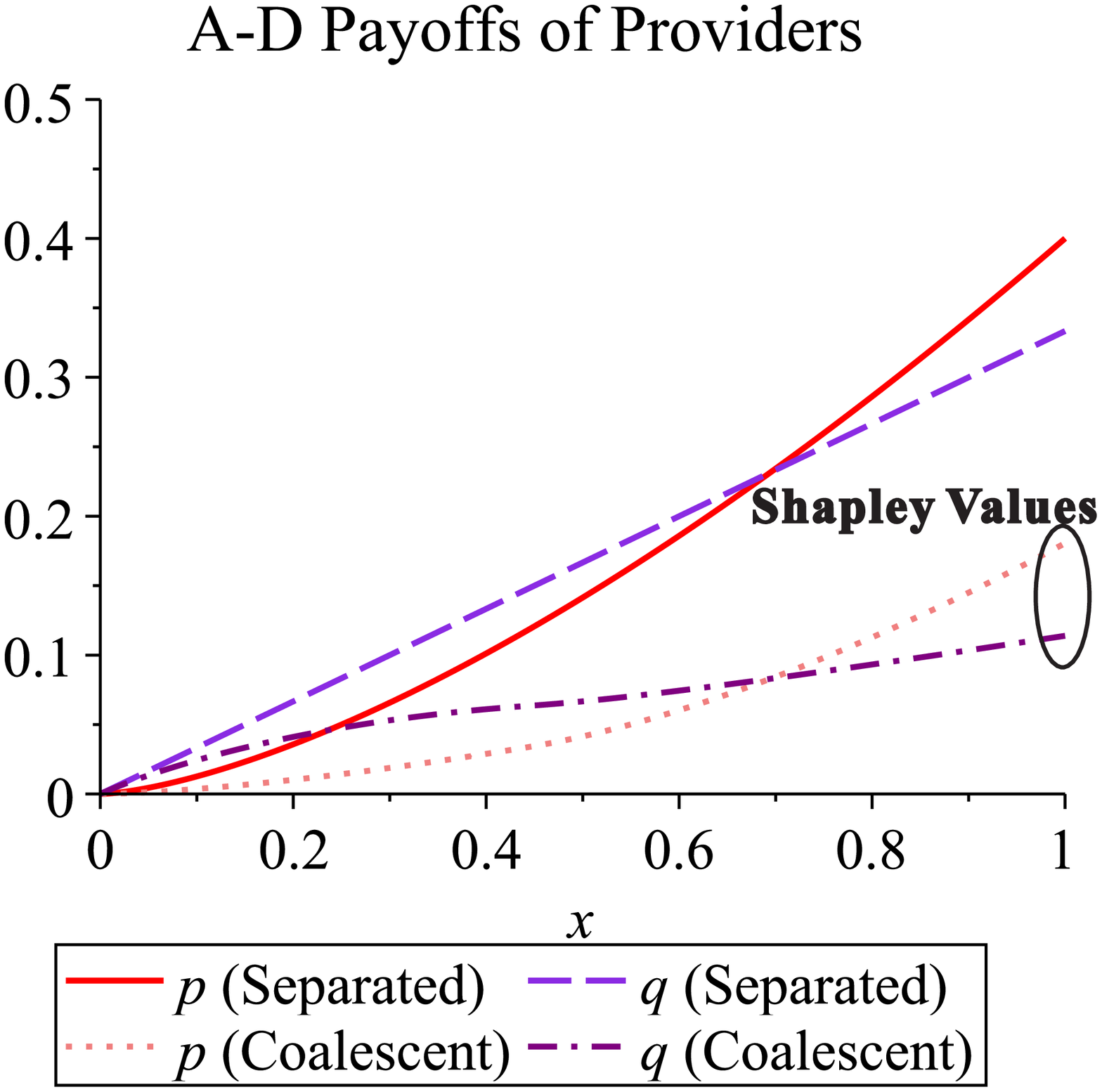}}
  \subfigure{\label{fig:ex2c}
        \includegraphics[width=4.25cm]{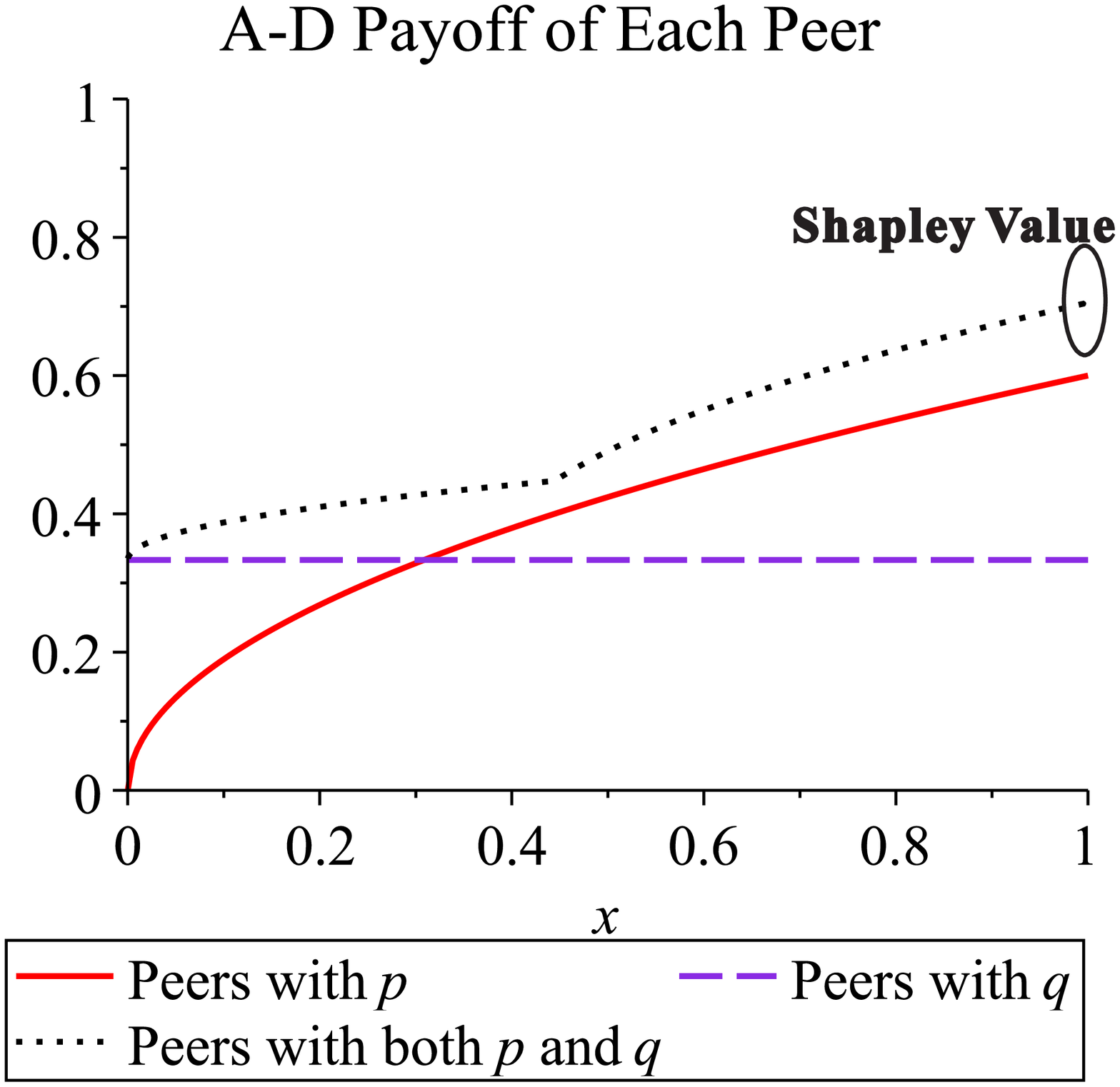}}
\caption{Example \ref{ex:monopoly}: A-D Payoffs of Two Providers and Peers for Concave Cost Functions.} \label{fig:ex2}
\end{center}
\end{figure*}

\section{Critique of A-D Payoff for Separate Providers}\label{sec:critique}

The discussion so far has focused on the stability of the grand
coalition.  The result in Theorem \ref{th:adnotcore} suggests that if
there is a noncontributing (free-riding) provider, which is true even
for concave cost functions for multiple providers, the grand coalition
will not be formed. The situation is aggravated by Theorem
\ref{th:convergencetocore}, stating that {\em single-provider
  coalitions} (\ie, the separated case) will persist if providers are
rational. We now illustrate the weak points of the A-D payoff under the
single-provider coalitions with three representative examples.

\subsection{Unfairness and Monopoly}\label{sec:fairness_monopoly}

\begin{ex}[Unfairness] \label{ex:unfairness} Suppose that there are two
  providers, \ie, ${Z}=\{p,q\}$, with ${\widetilde\Omega}_p
  (x)=7(x-1)^{1.5}/8+1/8 $ and ${\widetilde\Omega}_q (x)=1 - x $, both
  of which are decreasing and {\it
    convex}. 
  All values are shown in Fig. \ref{fig:ex1} as functions of $x$. In
  line with Theorem \ref{th:convergencetocore}, provider $p$ is paid more
  than her Shapley value, whereas peers are paid less than theirs.

  We can see that each peer $n$ will be paid $21/32$
  ($\widetilde{\varphi}_{n}^{\{p\}} (0)$) when he is contained by the
  coalition $\{p,n \}$ and the payoff decreases with the number of peers
  in this coalition.  On the other hand, provider $p$ wants to be
  assisted by as many peers as possible because
  $\widetilde{\varphi}_{p}^{\{p\}} (x)$ is increasing in $x$.  If it is
  possible for $n$ to prevent other peers from joining the coalition, he
  can get $21/32$.  However, it is more likely in real systems that no
  peer can kick out other peers, as discussed in \cite[Section
  5.1]{refMisraP2P} as well.  Thus, $p$ will be assisted by $x=0.6163$
  fraction of peers, which is the unique solution of
  $\widetilde{\varphi}_{n}^{\{p\}} (x)= \widetilde{\varphi}_{n}^{\{q\}}
  (x)$ while $q$ will be assisted by $1-x=0.3837$ fraction of peers.
\end{ex}

\smallskip
\begin{ex}[Monopoly]\label{ex:monopoly} Consider a two-provider system ${Z}=\{p,q\}$ with ${\widetilde\Omega}_p (x)=1 - x^{3/2} $ and ${\widetilde\Omega}_q (x)=1 - 2x/3 $, both of which are decreasing and {\it concave}.
Similar to Example \ref{ex:unfairness}, we can obtain $\widetilde{\varphi}_{p}^{\{p\}} (x) = 2x^{3/2}/5$, $\widetilde{\varphi}_{q}^{\{q\}} (x) = x/3$, $\widetilde{\varphi}_{n}^{\{p\}} (x) = 3x^{1/2}/5$ and $\widetilde{\varphi}_{n}^{\{q\}} (x) = 1/3$.
All values including the Shapley values are shown in Fig. \ref{fig:ex2}. Not to mention unfairness in line with Example \ref{ex:unfairness} and Theorem \ref{th:convergencetocore}, provider $p$ {\it monopolizes} the whole peer-assisted services. No provider has an incentive to cooperate with other provider. It can be seen that all peers will assist provider $p$ because $\widetilde{\varphi}_{n}^{\{p\}} (x) > \widetilde{\varphi}_{n}^{\{q\}} (x)$ for $x>25/81$. Appealing to Definition \ref{def:stability}, if the providers are initially separated, the coalition structure will converge to the service monopoly by $p$. In line with Lemma \ref{lem:concave} and Theorem \ref{th:adnotcore}, even if the grand coalition is supposed to be the initial condition, it is not stable in the sense of the core. The noncontributing provider (Definition \ref{def:noncontributing}) in this example is $q$.

\end{ex}

\subsection{Instability of A-D Payoff Mechanism}\label{sec:instability}

The last example illustrates the A-D payoff can even induce an analog of the limit cycle in nonlinear systems, \ie, a closed trajectory having the property that other trajectories spirals into it as time approaches infinity.



\smallskip
\begin{ex}[Oscillation]\label{ex:twotwo} Let us consider a game with two providers and two peers where $N= \{ p_1, p_2, n_1, n_2 \}$.
If $\{ n_1 \}$, $\{ n_2 \}$ and $\{n_1,n_2\}$ assist the content distribution of $p_1$,
the reduction of the distribution cost is respectively 10\$, 9\$ and 11\$ per month.
However, the hard disk maintenance cost incurred from a peer is 5\$.
In the meantime, if $\{ n_1 \}$, $\{ n_2 \}$ and $\{n_1,n_2\}$ assist the content distribution of $p_2$,
the reduction of the distribution cost is respectively 6\$, 3\$ and 13\$ per month.
In this case, the hard disk maintenance cost incurred from a peer is supposed to be 2\$ due to smaller contents of $p_2$ as opposed to those of $p_1$.

\end{ex}

\begin{figure}[t!]
\centering
\centerline{\includegraphics[width=8cm]{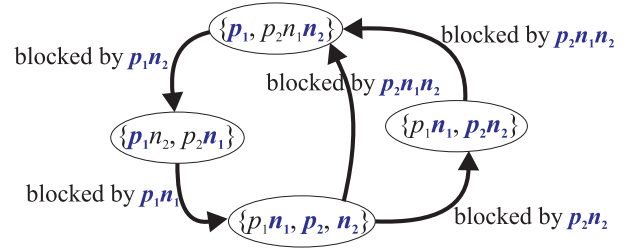}}
\caption{Example \ref{ex:twotwo}: A-D Payoff Leads to Oscillatory Coalition Structure.} \label{fig:oscillation}
\end{figure}

For simplicity, we omit the computation of
the A-D payoffs for all coalition
structures and stability analysis (see \ifthenelse{\boolean{arxiv}}{Appendix \ref{comp:ex:twotwo}}{Appendix of
  \cite{refJWYYGameNets} and Table 1 in \cite{refJWYYGameNets}} for details).
\ifthenelse{\boolean{arxiv}}{Table \ref{table:ad} contains A-D payoffs
  (and Shapley payoffs for the grand coalition) and blocking coalitions
  $C \subseteq N$ for any coalition structure where, for notational
  simplicity, we adopt a simplified expression for coalitional structure
  $\mathcal{P}$: a coalition $\{a, b , c\} \in \mathcal{P} $ is denoted
  by $ abc $ and each singleton set $\{i\}$ is denoted by $i$.  We first
  observe that the Shapley payoff of this example does not lie in the
  core.

  As time tends to infinity, the A-D payoff exhibits an oscillation of
  the partition $\mathcal{P}$ consisting of the four recurrent coalition
  structures as shown in Fig. \ref{fig:oscillation}.}{We first observe
  that the Shapley payoff of this example does not lie in the core. As
  time tends to infinity, the A-D payoff exhibits an oscillation of the
  partition $\mathcal{P}$ consisting of the four recurrent coalition
  structures as shown in Fig. \ref{fig:oscillation}, where, for
  notational simplicity, we adopt a simplified expression for
  coalitional structure $\mathcal{P}$: a coalition $\{a, b , c\} \in
  \mathcal{P} $ is denoted by $ abc $ and each singleton set $\{i\}$ is
  denoted by $i$. The evolution of coalition structure is governed by a simple rule: if there exist {\em blocking} coalitions (See Definition \ref{def:stability}), then arbitrary one of them will be formed.}

Let us begin with the partition $\{ p_1, p_2 n_1 n_2 \}$.
Player $p_1$ could have achieved the maximum payoff if he had formed a
coalition only with $n_1$.  However, player $n_1$ will remain in the
current coalition because he does not improve away from the current
coalition.  Instead, Player $n_2$ breaks the
coalition $p_2 n_1 n_2$ so that $n_2$ and $p_1$ can form coalition $p_1 n_2$ for their benefit. As
soon as the coalition $p_2 n_1 n_2$ is broken, $p_1$ betrays $n_2$ to
increase his payoff by colluding with $n_1$.
It is not clear how this behavior will be in large-scale systems, as
reported in the literature \cite{refTuticAD}.






  \setcounter{tempcounter}{\value{equation}}
\setcounter{equation}{0}
\renewcommand{\theequation}{FluidChi}
\begin{figure*}[t]
\begin{minipage}{\textwidth}
\begin{align}\label{eq:chigeneral}
\widetilde\chi_i^{\bar Z}(x)  & \textstyle = \widetilde\varphi_i^{Z}(1)  + \frac{w_i'}{x+\sum_{j\in \bar Z} w_j} \left( \sum_{j\in  \bar Z} \Omega_j(0) - M_\Omega^{\bar Z}(x) -  \left(x \widetilde\varphi_n^Z (1) + \sum_{j\in \bar Z} \widetilde\varphi_j^Z (1)\right) \right) ~ \mbox{where } w_i'= \left\{ \begin{array}{ll}
w_i, & \mbox{for } i \in \bar Z , \\
1, & \mbox{for } i \in \bar H ,
\end{array}
\right.
\end{align}
\end{minipage}
\put(-5,-19){\line(-1,0){508}}
\vspace{1mm}
\caption{Fluid $\chi$ payoff formula for multi-provider coalitions.}\label{fig:equation2}
\end{figure*}
\renewcommand{\theequation}{\arabic{equation}}
\setcounter{equation}{\value{tempcounter}}



\section{A Fair, Bargaining, and Stable Payoff Mechanism for Peer-Assisted Services}\label{sec:chivalue}

The key messages from the examples in Section~\ref{sec:critique} imply
that the A-D value of the separate case gives rise to unfairness,
monopoly, and even oscillation. Also, it turns out that some players'
coalition worth exceeds their Shapley payoffs which they are paid in the
grand coalition (Theorem \ref{th:adnotcore}). Thus, the Shapley payoff
scheme does not seem to be executable in practice because it is impossible
to make all players happy, unequivocally.  That being said, the fairness of
profit-sharing and the opportunism of players are difficult to stand together. Then, it is more reasonable to come up with a
compromising payoff mechanism that {\em (i)} forces players to {\it apportion}
the difference between the coalition worth and the sum of their fair shares,
{\em (ii)} grant providers a limited right of {\it bargaining}, and
{\em (iii)} {\it stabilize} the whole system. We will use a slightly different
notion of payoff mechanism, called $\chi$ value, originally proposed by
Casajus \cite{refChiValue}.

\subsection{An Axiomatic Characterization of $\chi$ Value}\label{sec:chiaxiomatic}

The $\chi$ value is characterized by a similar set of axioms
used for the A-D value. The only difference is that (i) {\slshape NP} is {\em weakened} to {\slshape GNP}, causing a deficiency in axiomatic characterization, which is made up by {\slshape WSP:}
\begin{axiom}[Grand Coalition Null Player, {\slshape GNP}]\label{ax:gnp}
If $v(K \cup \{ i \}) = v(K)$ for all $K \subseteq N$, then $ \phi_i (N,v,\{ N \}) =0 $.
\end{axiom}
\begin{axiom}[Weighted Splitting, {\slshape WSP}]\label{ax:wsp} If
  $\mathcal{P}'$ is finer than $\mathcal{P}$ (\ie, $C'(i) \subseteq
  C(i)$, $\forall i \in N$) and $j \in \mathcal{P}'(i)$,
\begin{align}
 \displaystyle  \frac{\phi_i (N,\!v ,\! \mathcal{P}) \!- \! \phi_i (N,\! v,\! \mathcal{P}')}{w_i}   \! = \!  \frac{\phi_j (N,\! v,\! \mathcal{P}) \!- \! \phi_j (N, \! v, \! \mathcal{P}')}{w_j}. \nonumber
 \end{align}
\end{axiom}
The cornerstone of $\chi$ value is the very observation that, as the grand coalition $\mathcal{P}=\{ N \}$ is broken into two or more coalitions, player $i$ now has another option to ally with other coalitions than $C(i) \in \mathcal{P} $ and this {\it outside option} must be assessed. To {\bf allow} the assessment of the outside options, it is inevitable to weaken {\slshape NP} (See Section \ref{sec:axiomatic}) to {\slshape GNP}, by satisfying only which, a player may receive positive payoff so far as he contributes to the worth of the grand coalition, even though he does not to that of the current coalition, \ie, {\slshape NP}. In the end, it is all about how to {\bf valuate} the outside option, the $\chi$ value's choice of which is to stick to the Shapley value by equally dividing the difference between the coalition worth and the sum of Shapley values, \ie, {\slshape WSP} for $\mathcal{P}=\{ N \}$.

Recalling the definition
$\varphi_{K} (N,v) = \sum_{i \in K} \varphi_i (N,v) $ in Section
\ref{sec:pri}, we present the following theorem (see
\cite{refChiValue,refCasajusNash} for the proof):
\begin{theorem}[$\chi$ Value]
  \label{thm:chivalue}
  The $\chi$ value is uniquely characterized
  by {\slshape CE}, {\slshape CS}, {\slshape ADD}, {\slshape GNP}
  and {\slshape WSP} as follows:
\begin{equation}\label{eq:chiformula}
\chi_i(N,v,\mathcal{P}) = \varphi_i  + \frac{w_i}{\sum_{k \in C(i)} w_k } \left( v(C(i)) - \varphi_{C(i)}\right)
\end{equation}
\end{theorem}
where $\varphi_i $ is Shapley value of player $i$ for non-partitioned
game $(N,v)=(N,v, \{ N \} )$.

\subsection{Fluid $\chi$ Value for Multi-Provider Coalitions}\label{sec:insensitivechi}

Recall $N= Z \cup H$, $\bar Z \subseteq Z$, $\bar H \subseteq H$ and $x
= |\bar H|/\eta$. To compute the $\chi$ payoff for the multiple provider
case, we first establish in the following theorem\footnotemark~a fluid
$\chi$ value in line with the analysis in Section
\ref{sec:insensitivead} with the limit axioms:
\begin{theorem}[$\chi$ Payoff for Multiple Providers]\label{th:chivaluecoal}
As $\eta$ tends to infinity, the $\chi$ payoffs of providers and peers under an arbitrary coalition $C=\bar Z \cup \bar H$ converge to \eqref{eq:chigeneral} in Fig. \ref{fig:equation2} where the Shapley payoffs $ \widetilde\varphi_i^{Z}(1)$ are given in \eqref{eq:admultivalue} in Fig. \ref{fig:equation}.
\end{theorem}

\footnotetext{In order to compute $\chi$ payoff of player $i$, we need
  to know not only the current coalition $C(i)$ but also Shapley values
  of players in $C(i)$. However, $\chi$ payoff still satisfies
  Definition \ref{def:ci}. Therefore, we can compute the payoff of
  player $i$ in coalition $C(i)$ irrespective of other coalitions.}

To intuitively interpret $\chi$ value, it is crucial to know the roles
of Axiom {\slshape WSP} and its weights $w_i$. In our context, because
of fairness between peers, it is more reasonable to set $w_i = 1$ for $i
\in H$. It does not make sense to differentiate payoffs between peers
due to the peer-homogeneity assumption in Section \ref{sec:worthpeer}.
On the contrary, we will clarify in Sections \ref{sec:chifair} and
\ref{sec:chibargaining} why the weights of providers $w_i $, $i \in Z$
do not necessarily have to be $1$. The essential difference between A-D
value and $\chi$ value lies in {\slshape WSP}.

\smallskip
\noindent {\bf\underline{Interpretation of WSP}}: It implies that, if
peer $i$ loses, say $\Delta_i$, when the coalition structure changes,
\eg, from the grand coalition $\mathcal{P} = \{ N \}$ to a finer
coalition structure $\mathcal{P}' \neq \{ N \}$, the provider $p \in
C(i)$ will lose $\Delta_i \times
w_p$. 
There are two implications of this weighted splitting. First, since the
payoff of each player $i$ is computed based on the baseline, \ie, the
Shapley value, and the surplus or deficit incurred by formation of the
coalition $C'(i)$ are {\em equally} distributed for $w_p=1$, $\chi$
value leads to a fair share of the profit. Secondly, now a provider may
{\em bargain} with peers over the dividend rate by setting $w_p$ to any
positive number. We elaborate on these two implications in the following
subsections.


\subsection{Fairness: Surplus-Sharing}\label{sec:chifair}

On the basis of the first implication of {\slshape WSP}, $\chi$ value
is {\it fairer} than A-D value in the following sense:
\begin{definition}[Surplus-Sharing]\label{def:pt}
A value $\phi$ of game $(N,v,\mathcal{P})$ is {\it surplus-sharing} if the following condition holds:
  if the coalition worth of coalition $C \in \mathcal {P}$ is greater than,
  equal to, or less than the sum of Shapley values of players in $C$,
  \ie, $\sum_{i \in C} \phi_i (N,v,\mathcal{P}) \gtreqqless \sum_{i \in
    C} \varphi_i(N,v)$, then the payoff of player $i \in C$ is greater
  than, equal to, or less than the Shapley value of player $i$,
  respectively, \ie, $ \phi_i (N,v,\mathcal{P}) \gtreqqless
  \varphi_i(N,v)$, for all $i \in C$ and for all $C \in \mathcal{P}$.
\end{definition}

Since we proved in Theorem \ref{th:convergencetocore} that, for $|Z|
\geq 2$, the payoff of provider $p$ in coalition $\{ p \} \cup H$
exceeds her Shapley value and that of peer $n \in H$ is smaller than
his, it is clear from this definition that A-D value is {\it not}
surplus-sharing for $|Z| \geq 2$, whereas $\chi$ value is
surplus-sharing for any $Z$, \eg, see \eqref{eq:chiformula} and
\eqref{eq:chigeneral}. For reference, both A-D and $\chi$ values are
surplus-sharing if $|Z|=1$.

\begin{figure}[t!]
\begin{center}
  \subfigure{\label{fig:cex1a}
        \includegraphics[width=4.25cm]{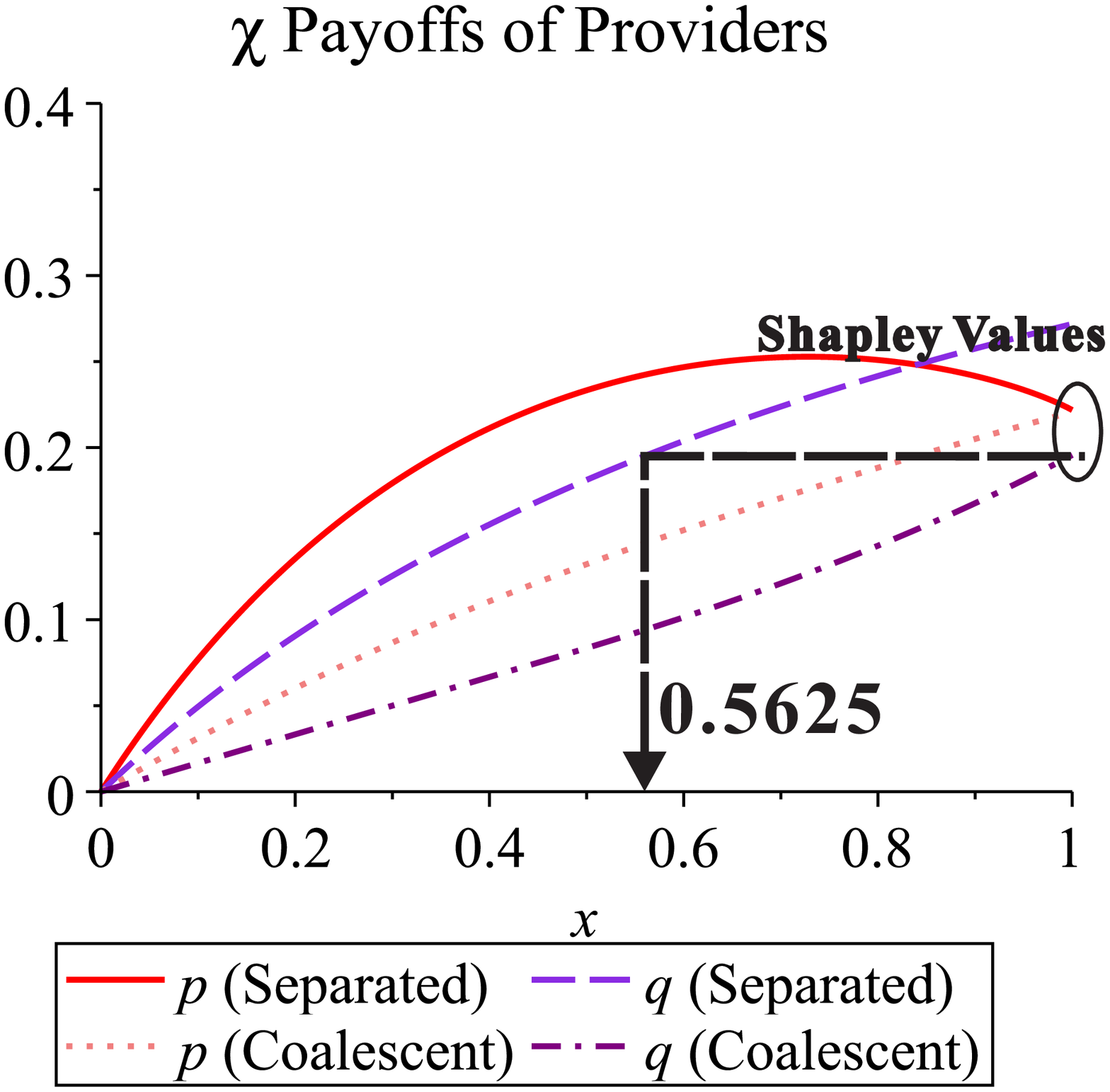}}
 \subfigure{\label{fig:cex1b}
        \includegraphics[width=4.25cm]{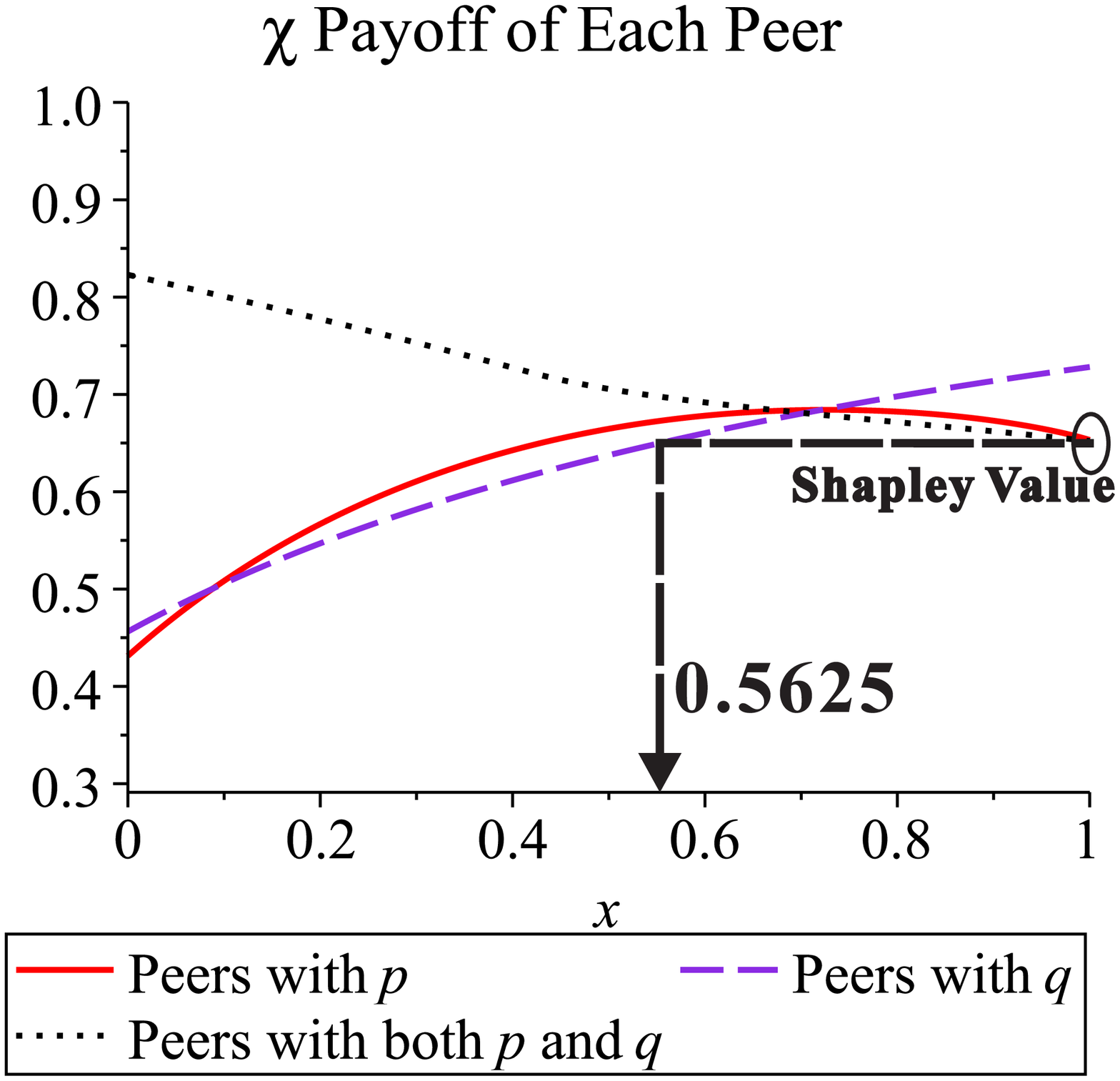}}
\caption{Example \ref{ex:unfairness}: $\chi$ Payoffs of Two Providers and Peers for Convex Cost Functions with $w_p=w_q=1$.} \label{fig:cex1}
\end{center}
\end{figure}

\begin{figure}[t!]
\begin{center}
 \subfigure{\label{fig:cex2a}
        \includegraphics[width=4.25cm]{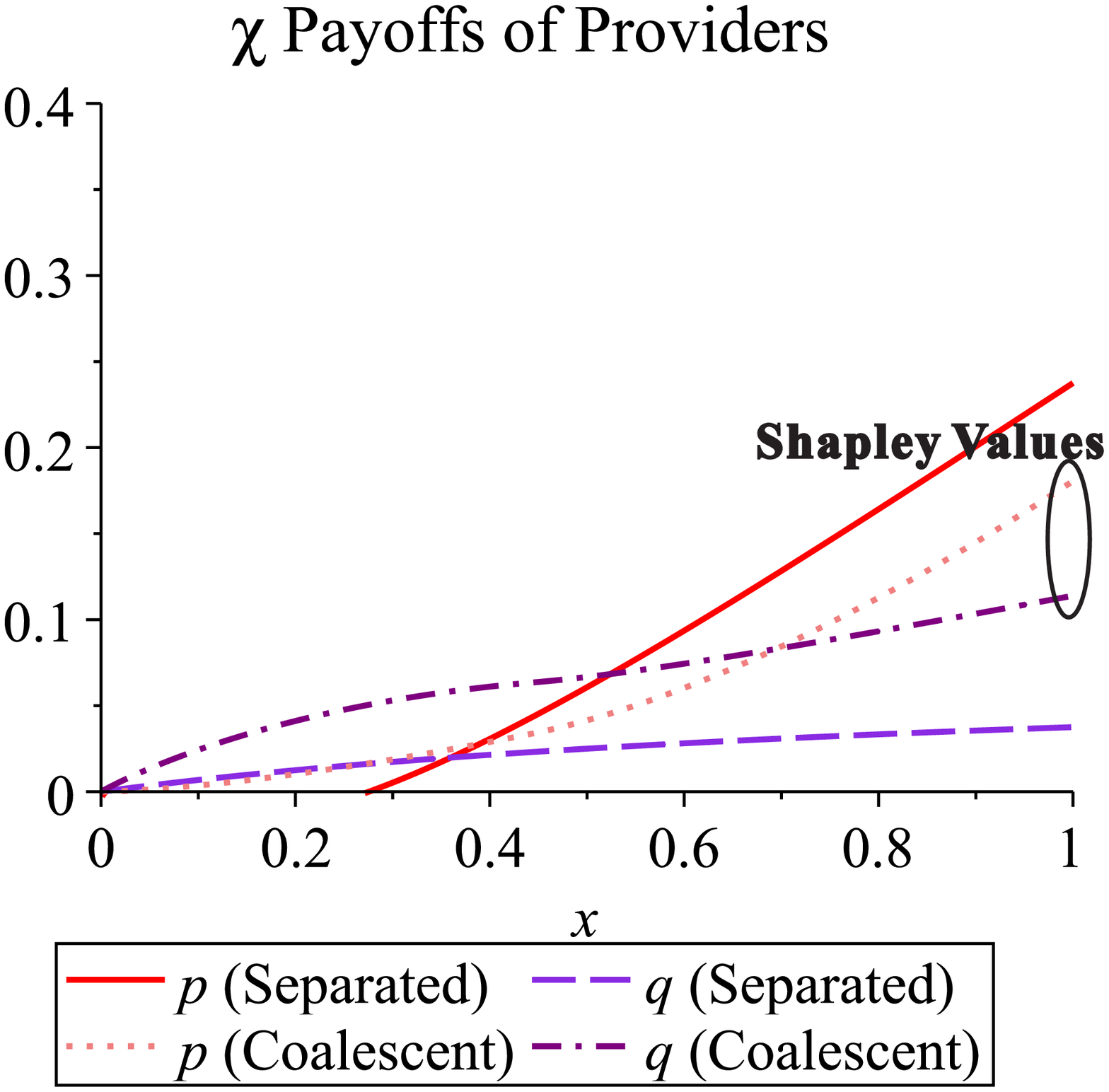}}
 \subfigure{\label{fig:cex2b}
        \includegraphics[width=4.25cm]{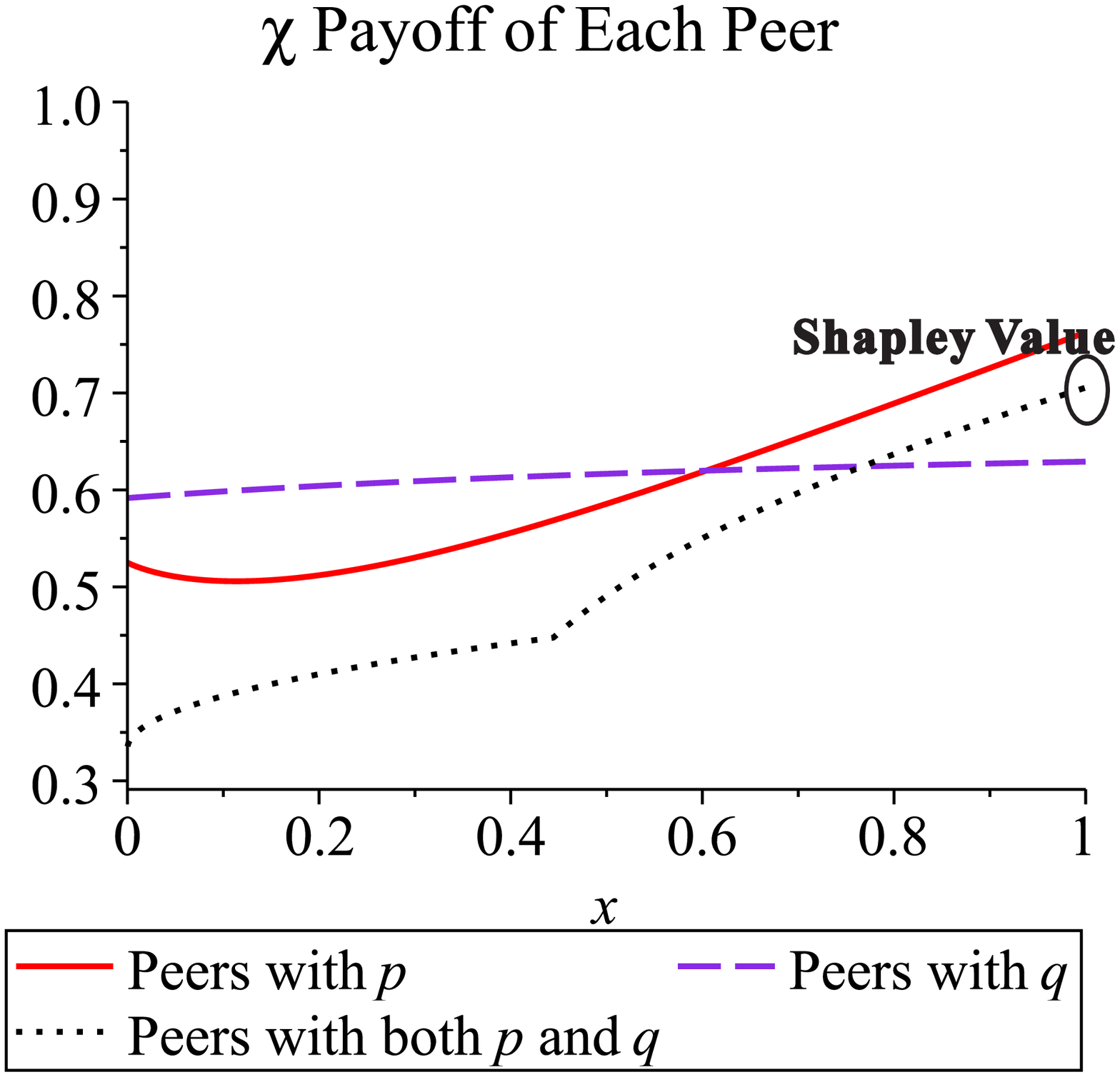}}
\caption{Example \ref{ex:monopoly}: $\chi$ Payoffs of Two Providers and Peers for Concave Cost Functions with $w_p=w_q=1$.} \label{fig:cex2}
\end{center}
\end{figure}

The corresponding $\chi$ payoffs of Examples \ref{ex:unfairness} and
\ref{ex:monopoly} for $w_i = 1$, $\forall i \in Z$, are shown in
Figs. \ref{fig:cex1} and \ref{fig:cex2}. As was the case of the A-D
payoffs in Examples \ref{ex:unfairness} and \ref{ex:monopoly}, the grand
coalitions are not stable. However, due to the surplus-sharing
property of the $\chi$ payoff, whenever the coalition worth is larger
than the Shapley sum of players in the coalition, {\it all} players in
the coalition are paid more and {\it vice versa}. For instance, we can
see from Fig. \ref{fig:cex1} that if the coalition is formed by provider
$q$ and $x>0.5625$ fraction of peers, all members of the coalition are
paid more than their respective Shapley payoffs.

As shown in Fig. \ref{fig:cex2}, the monopoly phenomenon of Example
\ref{ex:monopoly} for the case of A-D payoff is still observed for the
case of $\chi$ value. Regarding Example \ref{ex:unfairness}, as shown in Fig. \ref{fig:cex1}, $\chi$ payoff
even induces the monopoly by $q$, which did not
exist for the case of A-D payoff.

\setlength{\tabcolsep}{0pt}

\begin{table*}[t!]

\caption{Example \ref{ex:twotwo}: $\chi$ Payoff and Blocking Coalition $C$} \label{table:chi}
\small
\centering
\scalebox{0.80}{
\begin{tabular}{|c||c|c|c|c|c|}
\hline  & $\{ p_1 p_2, n_1 n_2 \}$ & $\{ p_1 p_2, n_1, n_2 \} $ & $\{ p_1 n_1, p_2, n_2 \} $ & $\{ p_1, p_2, n_1, n_2 \} $ & $\{ p_1 n_1,  p_2 n_2 \} $ \\
\hline $\chi_{p_1}$ & -1 & -1 & 5/3=1.67 & 0& 5/3=1.67 \\
\hline $\chi_{p_2}$ & 1 & 1 & 0 & 0  &  7/6=1.17  \\
\hline $\chi_{n_1}$ & 1/2=0.5 & 0 & 10/3=3.33 & 0& 10/3=3.33 \\
\hline $\chi_{n_2}$ & -1/2=-0.5 & 0 & 0 & 0 & -1/6=-0.17 \\
\hline \multirow{3}{*}{$C$} & \pink{$p_1$,$n_2$,$p_1 n_1$,$p_2 n_2$,$p_1 n_2$,$p_2 n_1$} & \pink{$p_1$,$p_1 n_1$,$p_1 n_2$,$p_2 n_1$} &  & \pink{$p_1 n_1$,$p_1 n_2$,$p_2 n_1$} &  \\
 & \pink{$p_1 p_2 n_1 n_2$,$p_1 p_2 n_1$,$p_1 p_2 n_2$} & \pink{$p_1 p_2 n_1 n_2$,$p_1 p_2 n_1$,$p_1 p_2 n_2$} & \mycyan{$-$} & \pink{$p_1 p_2 n_1 n_2$,$p_1 p_2 n_1$,$p_1 p_2 n_2$} & \pink{$n_2$} \\
  & \pink{$p_1 n_1 n_2$,$p_2 n_1 n_2$,$p_1 p_2 n_1 n_2$} & \pink{$p_1 n_1 n_2$,$p_2 n_1 n_2$,$p_1 p_2 n_1 n_2$} & & \pink{$p_2 n_1 n_2$,$p_1 p_2 n_1 n_2$} &  \\
\hline
\hline  & $\{ p_1 p_2 n_2, n_1 \} $   & $\{ p_1 p_2 n_1, n_2 \} $ & $\{ p_2 n_1 n_2, p_1 \} $ & $\{ p_1 p_2 n_1 n_2 \} $ & $\{ p_1 n_2, p_2, n_1 \} $  \\
\hline $\chi_{p_1}$  & 4/9=0.44 & 4/9=0.44 & 0 & 7/6 = 1.17 & 5/3=1.67 \\
\hline $\chi_{p_2}$  & 22/9=2.44 & 22/9=2.44 & 32/9=3.56 &  19/6 = 3.17 & 0 \\
\hline $\chi_{n_1}$  & 0 & 19/9=2.11 & 29/9=3.22 &  17/6 = 2.83 & 0 \\
\hline $\chi_{n_2}$  & 10/9=1.11 & 0  & 20/9=2.22 &  11/6 = 1.83 & 7/3=2.33 \\
\hline $C$ & \pink{$p_1 n_1 $,$p_1 p_2 n_1 n_2$,$p_2 n_1 n_2$,$p_1 n_2 $} & \pink{$p_1 n_1 $,$p_2 n_1 n_2$,$p_1 n_2$} & \pink{$p_1 n_2$, $p_1 n_1$}& \pink{$p_1 n_1$,$p_2 n_1 n_2$,$p_1 n_2$}  & \pink{$p_2 n_1$} \\
\hline
\hline  & $\{ p_1 n_1 n_2 , p_2 \} $  & $\{ p_1, p_2, n_1 n_2 \} $ & $\{ p_1, n_1, p_2 n_2 \}$ & $\{ p_1, n_2, p_2 n_1 \} $ & $\{ p_1 n_2 , p_2 n_1 \} $ \\
\hline $\chi_{p_1}$ & -4/9=-0.44  & 0 & 0       & 0 & 5/3=1.67\\
\hline $\chi_{p_2}$ & 0      & 0 & 7/6=1.17 & 13/6=2.17 & 13/6=2.17 \\
\hline $\chi_{n_1}$ & 11/9=1.22  & 1/2=0.5 & 0  &  11/6=1.83     & 11/6=1.83  \\
\hline $\chi_{n_2}$ & 2/9=0.22  & -1/2=-0.5 & -1/6=-0.17 & 0 & 7/3=2.33  \\
\hline \multirow{3}{*}{$C$}  & \pink{$p_1$,$p_1 n_1$,$p_1 n_2$,$p_2 n_1$,} & \pink{$n_2$,$p_1 n_1$,$p_2 n_2$,$p_1 n_2$,$p_2 n_1$} & \pink{$p_1 n_1$,$p_1 n_2 $,$p_2 n_1$,} &  \pink{$p_1 n_1$,$p_1 n_2$,} &  \\
&  \pink{$p_2 n_1 n_2$,$p_1 p_2 n_1 n_2$} & \pink{$p_1 p_2 n_1 n_2$,$p_1 p_2 n_1$,$p_1 p_2 n_2$}  & \pink{$p_1 p_2 n_1 n_2$,$p_1 p_2 n_1$}  & \pink{$p_1 p_2 n_1 n_2 $,$p_1 p_2 n_2$} & \mycyan{$-$} \\
&   \pink{$p_1 p_2 n_1$,$p_1 p_2 n_2$}  &  \pink{$p_2 n_1 n_2$,$p_1 p_2 n_1 n_2$} & \pink{$p_2 n_1 n_2 $,$p_1 p_2 n_2$,$n_2$}  &  \pink{$p_2 n_1 n_2 $,$p_1 p_2 n_1$} &  \\
\hline
\end{tabular}
}

\end{table*}

\subsection{Bargaining over the Dividend Rate}\label{sec:chibargaining}

Another implication of {\slshape WSP} is that a provider bargains with
peers over the division of the profit and loss by setting $w_i$ to a
nonnegative real value. For instance, consider the case when the
coalition worth exceeds the Shapley sum of players in the coalition,
\eg, $v(C(p)) > \varphi_{C(p)} $ in \eqref{eq:chiformula}, where $p\in
Z$ is the only provider in coalition $C(p)$. In this case, a provider
may award an extra bonus to peers by setting $w_p < 1$, or make more
profit by setting $w_p >1 $. For the coalition worth smaller than the
sum of Shapley payoffs, a provider may compensate peers
for loss by using $w_p >1$. Setting $w_p=1$ guarantees the fair
profit-sharing between provider $p$ and peers, whereas provider $p$ may
be willing to use $w_p \neq 1 $ for bargaining.

Although $w_p$ can be viewed as a {\it flexible knob} to balance the
fairness of the system and the bargaining powers of providers,
regulators need to control the providers by introducing upper and lower
bounds on $w_p$ which may depend on whether $v(C(p)) >
\varphi_{C(p)} $ or not, because $w_p$ have opposite meanings for the
two cases. For example, providers may use weights satisfying the
following condition:
\begin{eqnarray*}
\left\{ \begin{array}{ll}
w_p \geq \underline{w}_p , & \mbox{if } v(C(p)) < \varphi_{C(p)}, \\
0 \leq w_p \leq \overline{w}_p , & \mbox{if } v(C(p)) \geq \varphi_{C(p)} . \\
\end{array} \right.
\end{eqnarray*}
Two bounds, $\underline{w}_p$ and $\overline{w}_p$ can be viewed as a
preventive measure taken by the authorities to avoid {\it unfair
  rivalries} between providers.

\begin{figure}[t!]
\begin{center}
  \subfigure{\label{fig:wex1a}
        \includegraphics[width=4.25cm]{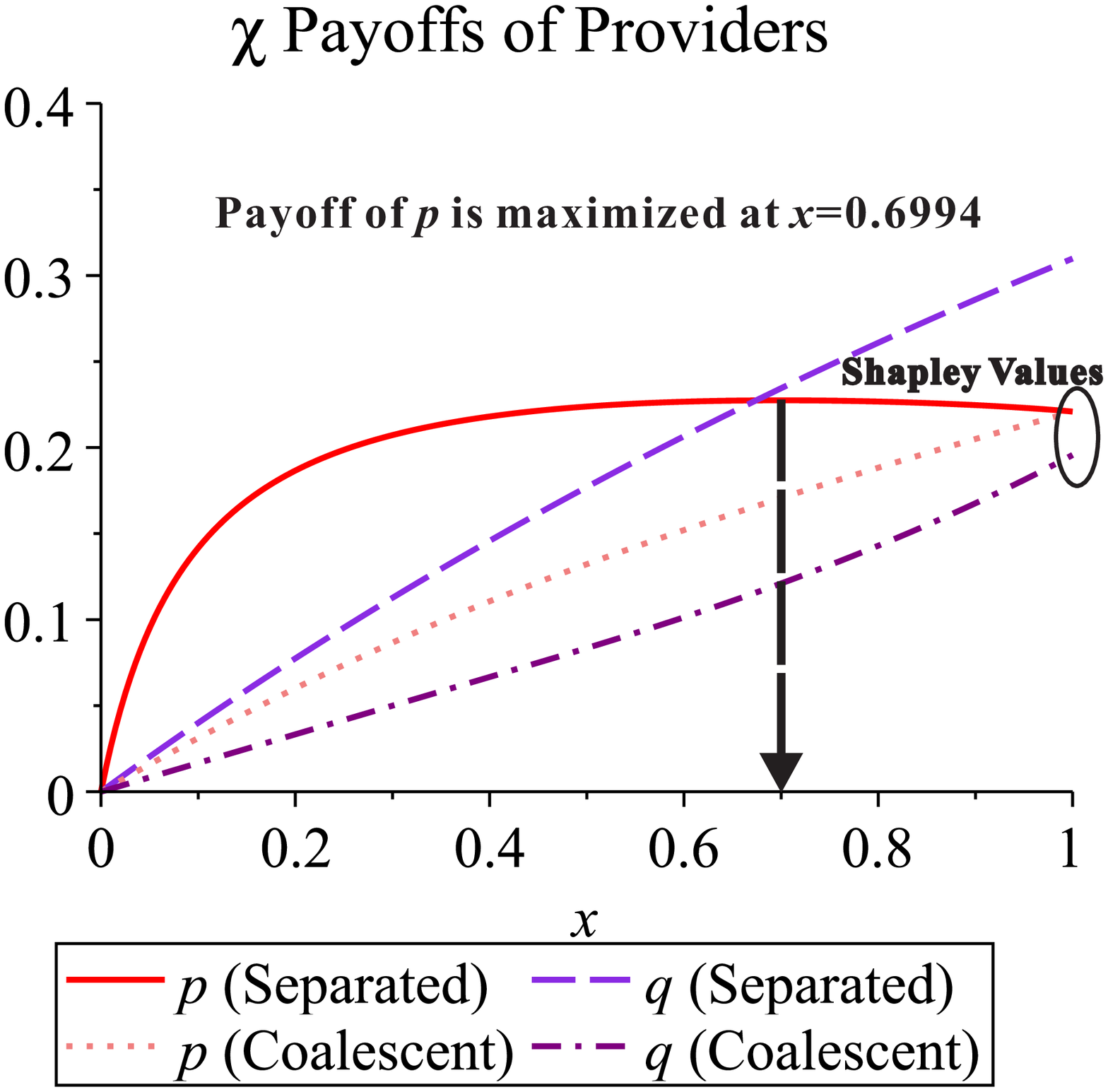}}
 \subfigure{\label{fig:wex1b}
        \includegraphics[width=4.25cm]{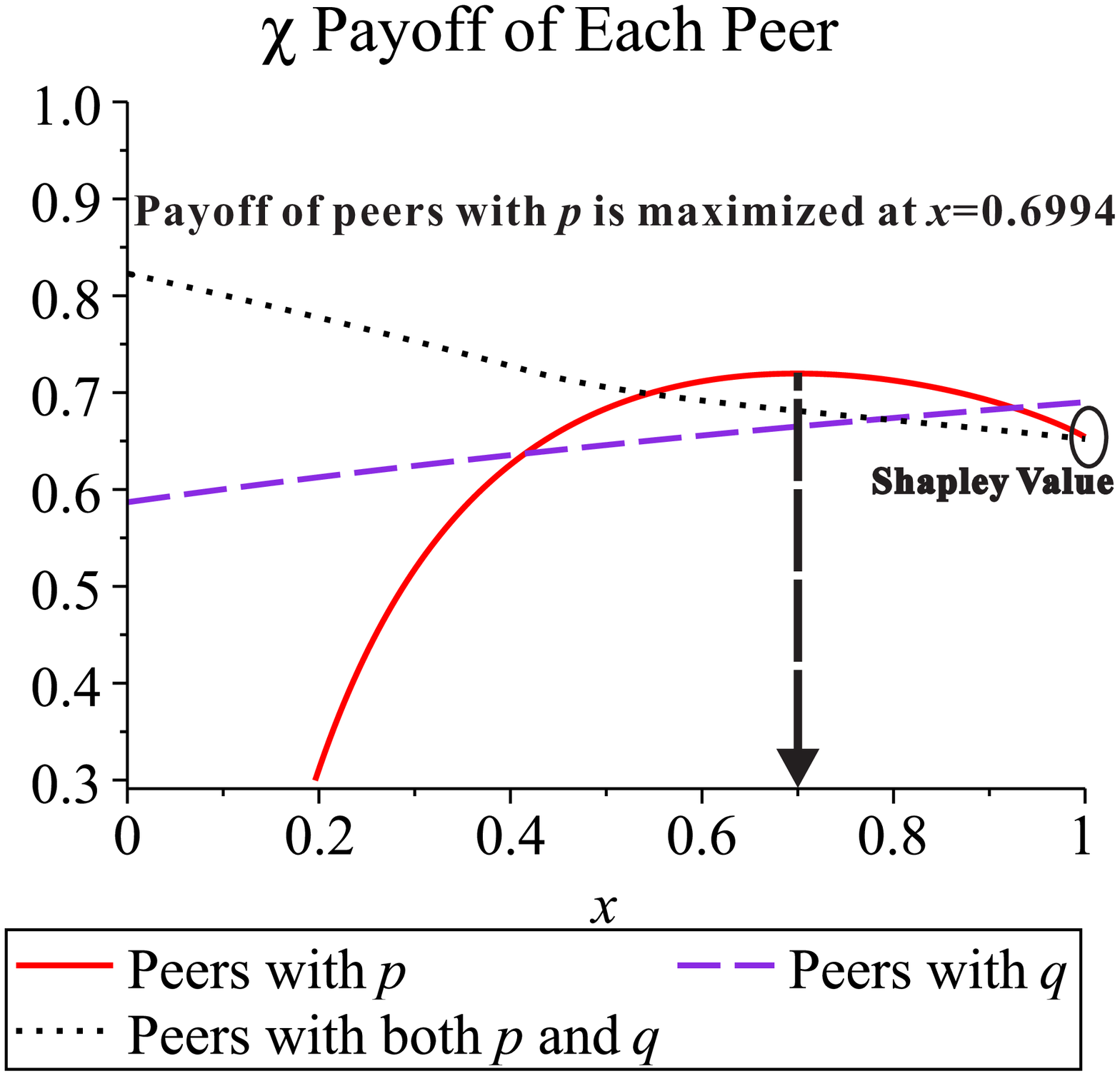}}
\caption{Example \ref{ex:unfairness}: $\chi$ Payoffs of Two Providers and Peers for Convex Cost Functions with $w_p=0.1$ and $w_q=3$.} \label{fig:wex1}
\end{center}
\end{figure}

Adopting non-identical weights $w_p=0.1$ and $w_q = 3$, we revisit
Example \ref{ex:unfairness}. Unlike Fig. \ref{fig:cex1} where provider
$q$ monopolizes all peers because $\widetilde\chi_q^{\{q\}}(1)$ and
$\widetilde\chi_i^{\{q\}}(1)$ for $i \in H$ is the biggest possible
payoffs for $q$ and any peers, the monopoly for this set of weights is
broken as shown in Fig. \ref{fig:wex1}. Now providers $p$ and $q$ will
possess $0.6994$ and $0.3006$ fraction of peers, respectively. It is
remarkable that the $\chi$ payoffs are still surplus-sharing as in
Figs. \ref{fig:cex1} and \ref{fig:cex2}.

\subsection{Stability of Coalition Structures}\label{sec:chistability}

The $\chi$ value of the game in Example \ref{ex:twotwo} with equal
weights $w_i =1$, for all $i \in N$, is shown in Table
\ref{table:chi}. As discussed in \cite{refChiValue}, {\slshape NP} is
not suitable for a value reflecting outside options. For example, let us
consider the partition $\{ p_1 p_2 , n_1, n_2 \}$. For the case of the
A-D value, payoffs of both providers $p_1$ and $p_2$ are $0$. However,
as we observe from Example \ref{ex:twotwo}, the best $p_1$ can do is
to ally with $n_1$ to reduce her operational cost by $v(\{ p_1, n_1\}) =
5$ whereas the best $p_2$ can do to reduce hers by $v(\{ p_2, n_1, n_2
\})=9$. In other words, $p_1$ should release $p_2$ so that $p_2$ can create her
worth because $p_2$ has a {\it worthier} outside option, to reflect which, $\chi$ value implementation ``punishes'' $p_1$ by giving
her a {\it negative} payoff $\chi_{p_1} = -1$.

We also observe from Table \ref{table:chi} that players who can be better off by leaving the current coalition are paid {\it more} than others. For example, consider the partition $\{ p_1 n_2, p_2, n_1 \}$. For the case of A-D payoffs, $p_1$ and $n_2$ received the same payoff $2$ (See \ifthenelse{\boolean{arxiv}}{Table \ref{table:ad}}{Table 1 in \cite{refJWYYGameNets}}). However, in Table \ref{table:chi}, $n_2$ is paid more than $p_1$ because $n_2$ has the potential for creating the worthiest coalition $p_1 p_2 n_1 n_2$ or $p_2 n_1 n_2$, \ie, $v(\cdot)= 9$. Though $n_2$ will not be able to break the partition $\{ p_1 n_2, p_2, n_1 \}$ according to the stability defined in Definition \ref{def:stability}, $n_2$ is paid more than $p_1$ essentially for its {\it assessed} potential. In this case, the final form of coalition structure after its endogenous evolution is the state $\{ p_1 n_2, p_2 n_1 \}$.
There are now two {\it absorbing} states $\{ p_1 n_1, p_2, n_2 \}$ and $\{ p_1 n_2, p_2 n_1 \}$, as shown in Table \ref{table:chi}, which are stable in the sense of Definition \ref{def:stability}. On the contrary, there does not exist any stable state for the case of A-D payoff as shown in Fig. \ref{fig:oscillation} (See also Section \ref{sec:instability} and \ifthenelse{\boolean{arxiv}}{Table \ref{table:ad}}{Table 1 in \cite{refJWYYGameNets}}).

A more general result \cite[Theorem 6.1]{refChiValue} is that, if we
adopt $\chi$ value to distribute the profit of the peer-assisted
services, the system always has at least one stable coalition structure,
irrespective of the number of providers. It it also remarkable that the
following theorem holds without any restriction on operational cost
$\widetilde\Omega_p(\cdot)$, whereas we assumed that
$\widetilde\Omega_p(\cdot)$ is non-increasing in
Section~\ref{sec:issuemultiple}.

\begin{theorem}[Stability of $\chi$ Payoff]\label{th:chistable} For $\chi$ value, there always exists a stable coalition structure $\mathcal{P}$.
\end{theorem}

Also, it follows from \cite[Corollary 6.4]{refChiValue} that the
instability of the grand coalition cannot be improved:
\begin{corollary}[Stability of Grand Coalition Preserved]\label{cor:corechi} The grand coalition of $\chi$ value is stable {\bf if and only if} the Shapley value lies in the core.
\end{corollary}

To summarize, even if we adopt $\chi$ value, the instability of the
grand coalition for the Shapley payoff which we observed in Theorem
\ref{th:adnotcore} remains {\it unchanged}. However, it is
guaranteed that there exists a stable coalition structure for $\chi$ value.


\section{Application to Delay-Tolerant Networks}\label{sec:app}

In this section, we present a concrete example of the peer-assisted
services in delay-tolerant networks where mobile users share certain
contents with each other in a peer-to-peer fashion \cite{refAgeGossip}:
whenever two mobile users meet, a user whose content is more recent
pushes it to the other whose content is outdated. We consider here a
single class case, using the method in \cite{refAgeGossip}.

We assume that there exist two providers, $p$ and $q$, whose contents
differ. Users who are subscribing
to the content of a provider are assumed to assist the provider in any case.  The fraction of
users subscribing to each provider is denoted by $x_p^0$ and $x_q^0$. As discussed in Section \ref{sec:worthpeer}, we
also assume that a non-subscribing user is allowed to assist at most one provider.
Suppose that the content providers $p$ and $q$ push content
updates to users, who are assisting providers, with the rate $\mu_p$ and
$\mu_q $, respectively, and every user meets other users with the
aggregate rate $\lambda$.
\begin{figure}[t!]
  \centering
  \centerline{\includegraphics[width=5cm]{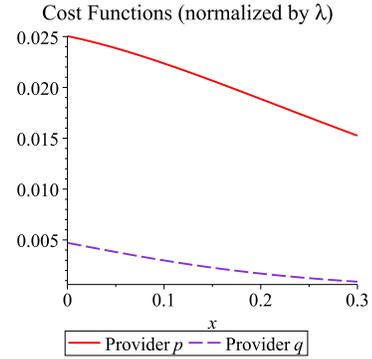}}
  \caption{Cost Functions of Two Providers in the Delay-Tolerant Network.} \label{fig:app}
\end{figure}
Then it follows from the analysis in \cite[Section 5.1]{refAgeGossip}
that, if $x_p \geq x_p^0 $ fraction of users are assisting provider $p$, for a user who is subscribing to
provider $p$, the expected age of the content and the outage probability that the age is larger than $G_p^{\rm max}$ are:
$$
\bar G_p = \frac{1}{x_p \lambda} \ln \frac{x_p \lambda + \mu_p}{\mu_p} ,\quad
P_p^C = \frac{{x_p \lambda + \mu_p }}{{x_p \lambda} + {\mu_p} {\rm e}^{(\mu_p + x_p \lambda) G_p^{\rm max} } }.
$$
The above two expressions can be easily derived by using integration by
parts.
A provider may guarantee subscribers a certain level of quality
of service by imposing constraints such as {\em (i)} $\bar G_p \leq 1
\textrm{min}$ or {\em (ii)} $ P_p^C \leq 0.01$ for $G_p^{\rm max} = 10
\textrm{min}$, of which we use the former here.

For instance, the cost function of provider $p$ can be computed by
solving the following optimization problem over $\mu_p$:
\begin{align}
\textstyle  \min_{\mu_p} x_p \mu_p \quad \textrm{subject to}: ~ \bar G_p \leq g_p \nonumber
\end{align}
where $x_p \mu_p $ corresponds to the average cost per user. The solution of this problem yields provider $p$'s cost function:
$$
{\widetilde\Omega}_p (x) \defeq x \mu_p^* = \frac{x^2 \lambda}{\exp\left( x \lambda g_p \right) - 1}
$$
where we dropped the subscript $p$ from $x_p$. Suppose $x_p^0 = 0.4$ and
$x_q^0 = 0.3$. If providers $p$ and $q$ use $g_p = 5 /{\lambda} $ and
$g_q = 10/{\lambda} $, \ie, provider $p$ has decided to maintain a lower
average age of the content than that of provider $q$, we get the cost
functions ${\widetilde\Omega}_p (x+x_p^0) / \lambda $ and
${\widetilde\Omega}_q (x+x_q^0) / \lambda $ as shown in
Fig. \ref{fig:app}. By computing the equations in \eqref{eq:addualvalue}
and \eqref{eq:chigeneral}, it is not difficult to see that provider $p$
monopolizes the remaining fraction of users, $1- x_p^0 - x_q^0 = 0.3$,
whether we adopt the A-D payoff or $\chi$ payoff. Nonetheless, users can
receive more under the $\chi$ payoff than under the A-D payoff due to
the surplus-sharing property discussed in Section \ref{sec:chifair}.


\section{Concluding Remarks and Future Work}\label{sec:conc}


A quote from an interview of BBC iPlayer with CNET UK \cite{refiPlayer}:
``{\em Some people didn't like their upload bandwidth being used. It was
  clearly a concern for us, and we want to make sure that everyone is
  happy, unequivocally, using iPlayer.}''

In this paper, we have first studied the incentive structure in
peer-assisted services with multiple providers, where the popular
Shapley value based scheme might be in conflict with the pursuit of
profits by rational content providers and peers.  The key messages from
our analysis are summarized as: First, even though it is fair to
pay peers more because they become relatively more useful as the number
of peer-assisted services increases, the content providers will not
admit that peers should receive their fair shares. The providers tend to
persist in single-provider coalitions. In the sense of the classical
stability notion, called `core', the cooperation would have been broken
even if we had begun with the grand coalition as the initial
condition. Second, we have illustrated yet another problems when we use
the Shapley-like incentive for the exclusive single-provider coalitions.
These results suggest that the profit-sharing system, Shapley value, and
hence its fairness axioms, are not compatible with the selfishness of
the content providers. We have proposed an alternate, realistic
incentive structure in peer-assisted services, called $\chi$ value,
which reflects a trade-off between fairness and rationality of
individuals. Moreover, the weights of $\chi$ value can serve as a
flexible knob to enable providers to bargain with peers over the
dividend rate {\it at the same time} as a preventive measure to avoid
cutthroat or unfair competition between providers.
However, we recognize the limitation of these results, which are based on the assumption that there is no additional cost reduction other than that achieved from the peer-partition optimization. We surmise that providers in cooperation can make further expenses cut by pooling and optimizing their resources, and traffic engineering, which will {\em transform} their cost functions. The question remains open how the ramifications of this type of cooperation can be quantified in peer-assisted services.

\section*{Acknowledgements}
This work was supported by 
KRCF (Korea Research Council of Fundamental Science and Research). The authors would like to thank the anonymous reviewers for their valuable comments to improve the quality of the paper.

\bibliographystyle{IEEEtran}
\bibliography{IEEEabrv,bib2}

\appendix
\section{Appendix}



\subsection{Proof of Theorem \ref{th:advaluecoal}}\label{proof:th:advaluecoal}
Recall that we use notation $\widetilde{\varphi}^{{\bar Z}} (x)$ to denote $\widetilde{\varphi} ({\bar Z} \cup \bar H, v)$. We use the mathematical induction to prove this theorem. The equation \eqref{eq:admultivalue} holds for $|{\bar Z}|=0$ and ${\bar Z} = \emptyset$ (empty set) because we have from \eqref{eq:admultivalue} that there is no provider to pay and $\widetilde{\varphi}_{n}^{\emptyset } (x) = 0 $ for $n\in \bar H$.

 Now we {\bf assume} that \eqref{eq:admultivalue} holds for all $\Xi \subsetneq {\bar Z}$ such that $|\Xi|\leq \xi$ where $\xi \geq 0$. To prove Theorem \ref{th:advaluecoal}, it {\it suffices} to show that \eqref{eq:admultivalue} also holds for all $\Xi'\subseteq {\bar Z}$ such that $|\Xi'|=\xi+1$. To this end, we first apply Axiom {\slshape CE}. As $\eta$ tends to infinity while $x$ remains unchanged, for $p \in \Xi'$ and $n \in \bar H$, Axiom {\slshape CE} for the partition $\{ \Xi' \cup \bar H \}$ can be rewritten as follows:
\begin{equation}\label{eq:ceapplq}
\textstyle \sum_{p \in \Xi'} \widetilde{\varphi}_{p}^{\Xi'} (x ) + x \widetilde{\varphi}_{n}^{\Xi'} (x) = \sum_{p\in \Xi'} {\widetilde\Omega}_p (0) - M_{\Omega}^{\Xi'} (x)
\end{equation}
which is the {\it normalized} (which we did in \eqref{eq:normalizecost}) total coalition worth created by the coalition $ \Xi' \cup \bar H $.
 Another axiom we apply is Axiom {\slshape FAIR} (fairness) which was used by Myerson \cite{refMyerson} to characterize the Shapley value. It follows from {\slshape FAIR} that
 \begin{align}\label{eq:fairnp}
\textstyle \widetilde{\varphi}_{n}^{\Xi'} (x) - \widetilde{\varphi}_{n}^{\Xi'\setminus \{ p \} } (x )= \frac{\ud}{\ud x} \widetilde{\varphi}_{p}^{\Xi'} ( x ) ,\quad \mbox{for all } p \in \Xi' .
\end{align}
Summing up \eqref{eq:fairnp} for all $p \in \Xi'$ and dividing the sum by $|\Xi'|=\xi+1$, we obtain
\begin{align}
 \widetilde{\varphi}_{n}^{\Xi'} (x) & \textstyle= \frac{1}{\xi+1}  \sum_{p \in \Xi'} \left( \widetilde{\varphi}_{n}^{\Xi'\setminus \{ p \} } (x )+\frac{\ud}{\ud x} \widetilde{\varphi}_{p}^{\Xi'} ( x ) \right) \nonumber \\
 & \textstyle= \frac{1}{\xi+1}  \sum_{p \in \Xi'} \widetilde{\varphi}_{n}^{\Xi'\setminus \{ p \} } (x )+ \frac{1}{\xi+1}  \frac{\ud}{\ud x} \sum_{p \in \Xi'} \widetilde{\varphi}_{p}^{\Xi'} ( x ) . \label{eq:peeradvaluetemp}
\end{align}
Plugging \eqref{eq:peeradvaluetemp} into \eqref{eq:ceapplq}, we obtain
\begin{align}
& \textstyle (\xi+1) \sum_{p \in \Xi'} \widetilde{\varphi}_{p}^{\Xi'} (x ) + {x} \frac{\ud}{\ud x} \sum_{p \in \Xi'} \widetilde{\varphi}_{p}^{\Xi'} ( x ) \nonumber \\
= & \textstyle (\xi+1) \left( \sum_{p\in \Xi'} {\widetilde\Omega}_p (0) - M_{\Omega}^{\Xi'} (x) \right)  - {x} \sum_{p \in \Xi'} \widetilde{\varphi}_{n}^{\Xi'\setminus \{ p \} } (x ). \label{eq:ceapplq2}
\end{align}
Since we know the form of $\widetilde{\varphi}_{n}^{\Xi'\setminus \{p\}} (x )$ for all $p \in \Xi'$ from the assumption ($\because |\Xi'\setminus\{ p \}|=\xi$), \eqref{eq:ceapplq2} is an ordinary differential equation of the function $\sum_{p \in \Xi'} \widetilde{\varphi}_{p}^{\Xi'} (x )$. Denote the RHS of \eqref{eq:ceapplq2} by $G(x)$. Appealing to \cite[Lemma 3]{refMisraP2P}, we get
\begin{align}
& \textstyle\sum_{p \in \Xi'} \widetilde{\varphi}_{p}^{\Xi'} (x )   = \int_{0}^{1} u^\xi G(u x) \ud u \nonumber \\
&\textstyle= \sum_{p\in \Xi'} {\widetilde\Omega}_p (0) - \int_{0}^{1} u^\xi (\xi+1) M_{\Omega}^{\Xi'} (ux) \ud u \nonumber \\
& \textstyle- \int_{0}^{1} u^{\xi+1} {x} \sum_{p \in \Xi'} \widetilde{\varphi}_{n}^{\Xi'\setminus \{ p \} } (u x ) \ud u , \nonumber\\
& \textstyle\frac{\ud}{\ud x} \sum_{p \in \Xi'} \widetilde{\varphi}_{p}^{\Xi'} (x ) =  - (\xi+1)\int_{0}^{1} u^{\xi+1}  \frac{\ud M_{\Omega}^{\Xi'}}{\ud x} (ux) \ud u \nonumber \\
& \textstyle- \int_{0}^{1} u^{\xi+1} \sum_{p \in \Xi'} \widetilde{\varphi}_{n}^{\Xi'\setminus \{ p \} } (u x ) \ud u \nonumber \\
& \textstyle -  \int_{0}^{1} u^{\xi+2} {x} \sum_{p \in \Xi'} \frac{\ud  \widetilde{\varphi}_{n}^{\Xi'\setminus \{ p \} }}{\ud x} (u x ) \ud u \label{eq:tempequ0} \\
 & \textstyle=  - (\xi+1)\int_{0}^{1} u^{\xi+1}  \frac{\ud M_{\Omega}^{\Xi'}}{\ud x} (ux) \ud u \nonumber \\
 & \textstyle + (\xi+1) \int_{0}^{1} u^{\xi+1} \sum_{p \in \Xi'} \widetilde{\varphi}_{n}^{\Xi'\setminus \{ p \} } (u x ) \ud u -  \sum_{p \in \Xi'} \widetilde{\varphi}_{n}^{\Xi'\setminus \{ p \} } .
 \label{eq:tempequ1}
\end{align}
where the last expression follows by integrating the last term of \eqref{eq:tempequ0} by parts. From \eqref{eq:peeradvaluetemp} and \eqref{eq:tempequ1}, $\widetilde{\varphi}_{n}^{\Xi'} (x)$ is rearranged as
 \begin{align}
\widetilde{\varphi}_{n}^{\Xi'} (x) & \textstyle =  - \int_{0}^{1} u^{\xi+1}  \frac{\ud M_{\Omega}^{\Xi'}}{\ud x} (ux) \ud u \nonumber \\
& \textstyle + \int_{0}^{1} u^{\xi+1} \sum_{p \in \Xi'} \widetilde{\varphi}_{n}^{\Xi'\setminus \{ p \} } (u x ) \ud u  . \label{eq:tempequ2}
\end{align}
From the assumption, $\widetilde{\varphi}_{n}^{\Xi'\setminus \{ p \} } (x )$ is given by \eqref{eq:admultivalue} for ${\bar Z} = \Xi'\setminus \{ p \}$, which is plugged into the last term of \eqref{eq:tempequ2} to yield
\begin{align}
& \textstyle\int_{0}^{1} u^{\xi+1} \sum_{p \in \Xi'} \widetilde{\varphi}_{n}^{\Xi'\setminus \{ p \} } (u x ) \ud u  \nonumber \\
& \textstyle = \!-\! \displaystyle \sum_{p\in \Xi'} \sum_{S \subseteq \Xi' \setminus \{p\}} \textstyle \int_{0}^{1} \! \int_{0}^{1} (ut)^{|S|} (u-ut)^{\xi-|S|} \frac{\ud M_{\Omega}^S}{\ud x} (utx) u \ud t \ud u. \label{eq:tempagain}
\end{align}
To reduce the double integral of \eqref{eq:tempagain}, we use the following fact:
\begin{align}
 & \textstyle \int_{0}^{1}  \int_{0}^{1}  (ut)^{|S|} (u-ut)^{\xi-|S|} f (utx) u\ud t \ud u \nonumber \\
 & \textstyle = \int_{0}^{1}  \int_{0}^{u}  \tau^{|S|} (u-\tau)^{\xi-|S|} f (\tau x) \ud \tau \ud u \nonumber \\
 & \textstyle = \int_{0}^{1}  \int_{\tau}^{1}  \tau^{|S|} (u-\tau)^{\xi-|S|} f (\tau x) \ud u \ud \tau \nonumber \\
& \textstyle =  \int_{0}^{1}  \frac{1}{\xi+1-|S|}  \tau^{|S|} (1-\tau)^{\xi+1-|S|} f (\tau x) \ud \tau \label{eq:tempagain2}
\end{align}
where we used the change of variable $ut=\tau$ and changed the order of the double integration with respect to $u$ and $\tau$. Plugging \eqref{eq:tempagain2} into \eqref{eq:tempagain} yields
\begin{align}
& \textstyle \int_{0}^{1} u^{\xi+1} \sum_{p \in \Xi'} \widetilde{\varphi}_{n}^{\Xi'\setminus \{ p \} } (u x ) \ud u \nonumber \\
 &= - \sum_{p\in \Xi'} \sum_{S \subseteq \Xi' \setminus \{p\}} \textstyle \frac{1}{\xi+1-|S|} \int_{0}^{1}   \scriptstyle u^{|S|} (1-u)^{\xi+1-|S|} \textstyle \frac{\ud M_{\Omega}^S}{\ud x} (ux) \ud u \nonumber \\
& \textstyle = -\sum_{S \subseteq \Xi' , S\neq \Xi'} \textstyle \int_{0}^{1}   u^{|S|} (1-u)^{\xi+1-|S|} \frac{\ud M_{\Omega}^S}{\ud x} (ux) \ud u . \label{eq:tempagain3}
\end{align}
where the last equality holds because
\begin{align}
 & \textstyle \sum_{p\in \Xi'} \sum_{S \subseteq \Xi' \setminus \{p\}} f( S) :  \sum_{S \subseteq \Xi',S\neq \Xi' } f( S) \nonumber \\
 & \textstyle =   (\xi+1) \cdot \binom{\xi}{|S|} : \binom{\xi+1}{|S|} = \xi+1-|S|:1 . \nonumber
 \end{align}
Plugging \eqref{eq:tempagain3} into \eqref{eq:tempequ2} establishes the following desired result:
 \begin{align}
\textstyle\widetilde{\varphi}_{n}^{\Xi'} (x)  =  -\sum_{S \subseteq \Xi'} \int_{0}^{1}   u^{|S|} (1-u)^{\xi+1-|S|} \frac{\ud M_{\Omega}^S}{\ud x} (ux) \ud u   \label{eq:peerest}
\end{align}
from which it follows
 \begin{align}
& \textstyle \widetilde{\varphi}_{n}^{\Xi'} (x) - \widetilde{\varphi}_{n}^{\Xi'\setminus \{ p \} } (x )  = \nonumber \\
& \textstyle  - \sum_{S \subseteq \Xi'} \int_{0}^{1}   u^{|S|} (1-u)^{\xi+1-|S|} \frac{\ud M_{\Omega}^S}{\ud x} (ux) \ud u \nonumber \\
& \textstyle + \sum_{S \subseteq \Xi' \setminus \{ p \} } \int_{0}^{1}   u^{|S|} (1-u)^{\xi-|S|} \frac{\ud M_{\Omega}^S}{\ud x} (ux) \ud u \nonumber .
\end{align}
The first term of the RHS can be decomposed into the following:
\begin{align}
& \textstyle-\sum_{S \subseteq \Xi'\setminus \{ p \} } \int_{0}^{1}   u^{|S|+1} (1-u)^{\xi+1-(|S|+1)} \frac{\ud M_{\Omega}^{S \cup \{ p \}}}{\ud x} (ux) \ud u \nonumber \\
& \textstyle- \sum_{S \subseteq \Xi'\setminus \{ p \} } \int_{0}^{1}   u^{|S|} (1-u)^{\xi+1-|S|} \frac{\ud M_{\Omega}^{S}}{\ud x} (ux) \ud u . \nonumber
\end{align}
Thus, we can obtain
 \begin{align}
& \textstyle \widetilde{\varphi}_{n}^{\Xi'} (x) - \widetilde{\varphi}_{n}^{\Xi'\setminus \{ p \} } (x ) = \nonumber \\
& \textstyle  \!-\! \displaystyle \sum_{S \subseteq \Xi' \setminus \{ p \} } \textstyle \int_{0}^{1}   \scriptstyle u^{|S|+1} (1-u)^{\xi-|S|} \textstyle \left( \frac{\ud M_{\Omega}^{S \cup \{ p \}}}{\ud x} \scriptstyle (ux) \textstyle - \frac{\ud M_{\Omega}^S}{\ud x} \scriptstyle (ux) \textstyle \right) \ud u \label{eq:idontknow} .
\end{align}
Integrating \eqref{eq:fairnp} with respect to $x$ and from \eqref{eq:idontknow}, we get
\begin{align}
&\textstyle \widetilde{\varphi}_{p}^{\Xi'} (x) = \nonumber \\
&- \sum_{S \subseteq \Xi' \setminus \{ p \} } \textstyle \int_{0}^{1}   \scriptstyle u^{|S|} (1-u)^{\xi-|S|} \textstyle \left[ M_{\Omega}^{S \cup \{ p \}} \scriptstyle (y) \textstyle - M_{\Omega}^S \scriptstyle (y) \textstyle  \right]^{y=ux}_{y=0} \ud u . \nonumber
\end{align}
Because $ M_{\Omega}^{S \cup \{ p \}}  (0) \textstyle - M_{\Omega}^S  (0) = {\widetilde\Omega}_p(0) $, the above equation combined with \eqref{eq:peerest} finally establishes that \eqref{eq:admultivalue} also holds for all $\Xi'\subseteq {\bar Z}$ where $|\Xi'|=\xi+1$, hence completing the proof.

\subsection{Proof of Theorem \ref{th:adnotcore}}\label{proof:th:adnotcore}

To prove the theorem, we need to show that the condition for the core in Definition \ref{def:core} is violated, implying that it suffices to show the following:
\begin{align}
\textstyle \widetilde\varphi_p^{ Z}(1) & \textstyle > \displaystyle \sum_{i\in { Z}} \textstyle {\widetilde\Omega}_i(0) \!- \! M_{\Omega}^{Z}(1) \!-\! \left( \! \displaystyle \sum_{i\in { Z} \setminus \{p\} } \textstyle {\widetilde\Omega}_i(0) - M_{\Omega}^{{ Z}\setminus \{p \}}(1) \! \right) \nonumber \\
& \textstyle = {\widetilde\Omega}_p(0) - \left( M_{\Omega}^{ Z}(1) - M_{\Omega}^{{ Z}\setminus \{p \}}(1) \right) .\label{eq:thethird}
\end{align}
This means that the payoff of $p \in { Z}$ is greater than the marginal increase of the limit worth, \ie,
$$\textstyle \lim_{\eta \to \infty} \frac{1}{\eta}v({ Z}\cup  H) - \lim_{\eta \to \infty} \frac{1}{\eta}v(({ Z}\setminus\{p\})\cup  H) .$$
Subtracting the RHS of \eqref{eq:thethird} from the LHS of \eqref{eq:thethird} and using the expression of $\widetilde\varphi_p^{ Z}(1)$ in  \eqref{eq:admultivalue}, we have
\begin{align}
& \textstyle M_{\Omega}^{ Z}(1) - M_{\Omega}^{{ Z}\setminus \{p \}}(1) \nonumber \\
& \textstyle - \displaystyle \sum_{S \subseteq { Z} \setminus \{p\}} \textstyle \int_{0}^{1} u^{|S|} (1-u)^{|{ Z}|-1-|S|} \left( M_{\Omega}^{S \cup \{p\}} \scriptstyle (u) \textstyle- M_{\Omega}^{S } \scriptstyle (u) \textstyle \right)  \ud u . \label{eq:toshowadnotcore}
\end{align}
We see from Definition \ref{def:noncontributing} that $M_{\Omega}^{ Z}(1) - M_{\Omega}^{{ Z}\setminus \{ p \}}(1) = {\widetilde\Omega}_{p}(0)$. From the assumption, there exists a noncontributing provider which we denote by $p$.
To show that \eqref{eq:toshowadnotcore} is strictly positive, we rewrite the last factor of the integrand as follows:
\begin{align}
 & \textstyle M_{\Omega}^{S \cup \{ p \}} (y) - M_{\Omega}^{S} (y) = \nonumber \\
 & \textstyle  \min \left\{  \sum_{i\in S \cup \{ p \}} {\widetilde\Omega}_i (y_i) ~\big\vert~ \sum_{i\in S \cup \{ p \} } y_i \leq y , y_i \geq 0 \right\} \nonumber \\
 & \textstyle - \min \left\{  \sum_{i\in S  } {\widetilde\Omega}_i (y_i) ~\big\vert~ \sum_{i\in S} y_i \leq y , y_i \geq 0 \right\}  \nonumber
\end{align}
where the first term in the RHS can be rearranged as
\begin{align}
&\textstyle \min \left\{  \sum_{i\in S \cup \{ p \}} {\widetilde\Omega}_i (y_i) ~\big\vert~ \sum_{i\in S \cup \{ p \} } y_i \leq y , y_i \geq 0 \right\}   \nonumber \\
& \textstyle \leq  {\widetilde\Omega}_p (0) + \min \left\{  \sum_{i\in S } {\widetilde\Omega}_i (y_i) ~\big\vert~ \sum_{i\in S } y_i \leq y , y_i \geq 0 \right\} \nonumber
\end{align}
where the inequality holds from that ${\widetilde\Omega}_i (y)$, $i \in { Z}$, are non-increasing. It can be easily seen that the inequality holds by considering two cases $y_p=0$ and $y_p>0$. The inequality becomes {\it strict} when $S=\emptyset$ over some interval in $[0,x]$ whose length is positive due to the assumption that ${\widetilde\Omega}_p (y) $ is not constant in the interval $y \in [0, x]$ and non-increasing. From this, we see that \eqref{eq:toshowadnotcore} is greater than
\begin{align}\nonumber
 {\widetilde\Omega}_{p}(0) - \sum_{S \subseteq { Z} \setminus \{p\}} \int_{0}^{1} u^{|S|} (1-u)^{|{ Z}|-1-|S|} {\widetilde\Omega}_p(0)  \ud u =0
\end{align}
which establishes \eqref{eq:thethird}, hence completing the proof.

\subsection{Proof of Theorem \ref{th:convergencetocore}}\label{proof:th:convergencetocore}
To prove Theorem \ref{th:convergencetocore}, it suffices to show that the following is positive for $\{p \}  \subsetneq T$ such that $T\subseteq { Z}$:
\begin{align}
& \textstyle \varphi_p^{\{ p \} } (x) - \varphi_p^{ T } (x) = -  \int_{0}^{1} M_{\Omega}^{\{p\}} \scriptstyle (ux) \textstyle   \ud u  \nonumber \\
& \displaystyle \sum_{S \subseteq T \setminus \{p\}} \textstyle \int_{0}^{1} u^{|S|} (1-u)^{|T|-1-|S|} \left( M_{\Omega}^{S \cup \{p\}} \scriptstyle (ux) \textstyle- M_{\Omega}^{S } \scriptstyle (ux) \textstyle \right)  \ud u \label{eq:showthisth2}
\end{align}
which implies that the payoff of $p$ when it is the only provider of the coalition is larger than that with other providers $ T \setminus \{ p \} $. To this end, we first observe that, for $y\leq x$,
\begin{align}
& \textstyle M_{\Omega}^{S \cup \{ p \}} (y) - M_{\Omega}^{S} (y) = \nonumber \\
& \textstyle  \min \left\{  \sum_{i\in S \cup \{ p \}} {\widetilde\Omega}_i (y_i) ~\big\vert~ \sum_{i\in S \cup \{ p \} } y_i \leq y , y_i \geq 0 \right\} \nonumber \\
& \textstyle - \min \left\{  \sum_{i\in S  } {\widetilde\Omega}_i (y_i) ~\big\vert~ \sum_{i\in S} y_i \leq y , y_i \geq 0 \right\} . \nonumber
\end{align}
Here the first term in the RHS can be rearranged as
\begin{align}
& \textstyle \min \left\{  \sum_{i\in S \cup \{ p \}} {\widetilde\Omega}_i (y_i) ~\big\vert~ \sum_{i\in S \cup \{ p \} } y_i \leq y , y_i \geq 0 \right\} \nonumber \\
& \textstyle \geq   M_{\Omega}^{\{p\}} (y)  + \min \left\{  \sum_{i\in S } {\widetilde\Omega}_i (y_i) ~\big\vert~ \sum_{i\in S } y_i \leq y , y_i \geq 0 \right\} \nonumber
\end{align}
where the inequality holds from that $M_{\Omega}^{\{i\}}(y)$, $i \in T$, are non-increasing. It can be easily seen that the inequality holds by considering two cases $y_p=0$ and $y_p>0$. The inequality becomes {\it strict} when $S=\emptyset$ over some interval in $[0,x]$ whose length is positive due to the assumption that ${\widetilde\Omega}_p (y) $ is not constant in the interval $y \in [0, x]$ and non-increasing. From this inequality, we have $M_{\Omega}^{S \cup \{ p \}} (y) - M_{\Omega}^{S} (y) \geq M_{\Omega}^{\{p\}} (y)$ and the inequality is strict over some interval of positive length. Plugging this relation into \eqref{eq:showthisth2} yields $\varphi_p^{\{ p \} } (x) - \varphi_p^{ T } (x)  > 0 $.

Note that from \eqref{eq:ceapplq}, we have:
\begin{align}
& \textstyle  \lim_{\eta \to\infty} v(\{p\} \cup \bar H) /\eta \textstyle = {\widetilde\Omega}_p(0) - M_{\Omega}^{\{p\}} (x) \nonumber \\
 & \textstyle \leq \sum_{i\in T} {\widetilde\Omega}_i (0) - M_{\Omega}^{T} (x) = \lim_{\eta \to\infty} v(T \cup \bar H)/\eta \nonumber
 \end{align}
which, when combined with $\varphi_p^{\{ p \} } (x) > \varphi_p^{ T } (x)$, implies the second part of the theorem.

\ifthenelse{\boolean{arxiv}}{\subsection{Computation of the A-D Payoff in Example \ref{ex:twotwo}}\label{comp:ex:twotwo}

From the description of the cost reduction and the hard disk maintenance cost incurred from peers, we can compute the {\it net} cost reduction for all possible coalitions.
For example, if $n_1 $ and $n_2$ help $p_1$, the coalition worth becomes $ v(\{n_1, n_2, p_1\}) = 11 \$ - 5 \$ - 5  \$ = 1 \$ . $
In a similar way, we have the following result:
$$\hat v(S)= \left\{ \begin{array}{ll}
0, & \mbox{if } S \mbox{ is not profitable}, \\
5, & \mbox{if } S=\{ p_1, n_1\} , \\
4, & \mbox{if } S=\{ p_1, n_2\}, \\
1, & \mbox{if } S=\{ p_1, n_1, n_2\} ,\\
4, & \mbox{if } S=\{ p_2, n_1\}, \\
1, & \mbox{if } S=\{ p_2, n_2\}, \\
9, & \mbox{if } S=\{ p_2, n_1, n_2 \}. \\
\end{array} \right.
$$
Using the same expression \eqref{eq:worthmultiple} as in Section \ref{sec:issuemultiple}, it is easy to see that the coalition worths for coalescent provider cases are $v(\{ p_1, p_2, n_1 \}) = 5$, $v(\{ p_1, p_2, n_2 \}) = 4$ and $v(\{ p_1, p_2, n_1, n_2 \}) = 9$.

 Suppose that peers continue to form a new coalition $C'$ when they can improve away from the current coalition $C$. That is, if there is a blocking coalition $C'$ (Definition \ref{def:stability}), they will betray $C$. It is easy to see that almost all coalition structures are {\it transient}. Note that, as discussed in \cite{refHartEndogenous}, one can get two types of resulting coalition structures when a player departs from a coalition. For example, if player $ n_1 $ departs from her coalition in  $ \{ p_1 n_1 n_2, p_2 \} $ to form a coalition with $ p_2 $, we may get either $ \{ p_1,  n_2, p_2 n_1 \} $ or $ \{ p_1  n_2, p_2 n_1 \} $. To obtain Table \ref{table:ad}, we assumed only the latter case to simplify the exposition.

\setlength{\tabcolsep}{1pt}

\begin{table*}[t!]
\centering \small \caption{Example \ref{ex:twotwo}: A-D Payoff and Blocking Coalition $C$} \label{table:ad}
\begin{tabular}{|c||c|c|c|c|c|}
\hline  & $\{ p_1 p_2, n_1 n_2 \}$ & $\{ p_1 p_2, n_1, n_2 \} $ & $\{ p_1, p_2, n_1 n_2 \} $ & $\{ p_1, p_2, n_1, n_2 \} $ & $\{ p_1 n_1,  p_2 n_2 \} $ \\
\hline $\varphi_{p_1}$ & 0 & 0 & 0 & 0& 5/2=2.5 \\
\hline $\varphi_{p_2}$ & 0 & 0 & 0 & 0  &  1/2=0.5  \\
\hline $\varphi_{n_1}$ & 0 & 0 & 0 & 0& 5/2=2.5 \\
\hline $\varphi_{n_2}$ & 0 & 0 & 0 & 0 & 1/2=0.5 \\
\hline $C$ & \multicolumn{4}{c|}{\pink{$p_1 n_1$,$p_1 n_2$,$p_2 n_1$,$p_2 n_2$,$p_1 p_2 n_1 n_2$,$p_2 n_1 n_2$}} & \pink{$p_2 n_1 n_2 $} \\
\hline recurrent & \bfseries{X} & \bfseries{X} & \bfseries{X} & \bfseries{X} & \bfseries{O} \\
\hline
\hline  & $\{ p_1 p_2 n_1 n_2 \} $  & $\{ p_1 p_2 n_1, n_2 \} $ & $\{ p_2 n_1 n_2, p_1 \} $ & $\{ p_1 p_2 n_2, n_1 \} $ & $\{ p_1 n_2, p_2, n_1 \} $  \\
\hline $\varphi_{p_1}$  & 7/6 = 1.17 & 7/6=1.17 & 0 & 5/3=1.67 & 2 \\
\hline $\varphi_{p_2}$  & 19/6 = 3.17 & 2/3=0.67 & 23/6=3.83 &  1/6=0.17 & 0 \\
\hline $\varphi_{n_1}$  & 17/6 = 2.83 & 19/6=3.17 & 10/3=3.33 & 0 & 0 \\
\hline $\varphi_{n_2}$  & 11/6 = 1.83 & 0  & 11/6=1.83 & 13/6=2.17 & 2 \\
\hline $C$ & \multicolumn{3}{c|}{ \pink{$p_1 n_2$} } & \multicolumn{2}{c|}{ \pink{$p_1 n_1 $,$p_2 n_1$}}   \\
\hline recurrent & \bfseries{X} & \bfseries{X} & \bfseries{O} & \bfseries{X} & \bfseries{X} \\
\hline
\hline  & $\{ p_1 n_1 n_2 , p_2 \} $  & $\{ p_1 n_1, p_2, n_2 \} $ & $\{ p_1, n_1, p_2 n_2 \}$ & $\{ p_1 n_2 , p_2 n_1 \} $ &  $\{ p_1, n_2, p_2 n_1 \} $\\
\hline $\varphi_{p_1}$ & 11/6=1.83  & 5/2=2.5 & 0       & 2 &  0\\
\hline $\varphi_{p_2}$ & 0      & 0 & 1/2=0.5 & 2 &  2\\
\hline $\varphi_{n_1}$ & -1/6=-0.17  & 5/2=2.5 & 0       & 2 &  2 \\
\hline $\varphi_{n_2}$ & -2/3=-0.67  & 0& 1/2=0.5 & 2 &  0 \\
\hline \multirow{3}{*}{$C$}  & \pink{$p_1 n_1$,$p_1 n_2$,$p_2 n_1$,} & & \pink{$p_1 n_1$,$p_1 n_2 $,$p_2 n_1$,} &  & \pink{$p_1 n_1$,$p_1 n_2$,} \\
&  \pink{$p_2 n_2$,$p_2 n_1 n_2$,} & \pink{$p_2 n_1 n_2 $,$p_2 n_2$}  & \pink{$p_1 p_2 n_1 n_2$,}  & \pink{$p_1 n_1$} & \pink{$p_1 p_2 n_1 n_2 $} \\
&   \pink{$n_1$, $n_2$, $n_1 n_2 $}  &   & \pink{$p_2 n_1 n_2 $}  &  & \pink{$p_2 n_1 n_2 $} \\
\hline recurrent & \bfseries{X} & \bfseries{O} & \bfseries{X} & \bfseries{O} & \bfseries{X} \\
\hline
\end{tabular}
\end{table*}
}{}

\begin{IEEEbiography}[{\includegraphics[width=1in,height=1.25in,keepaspectratio]{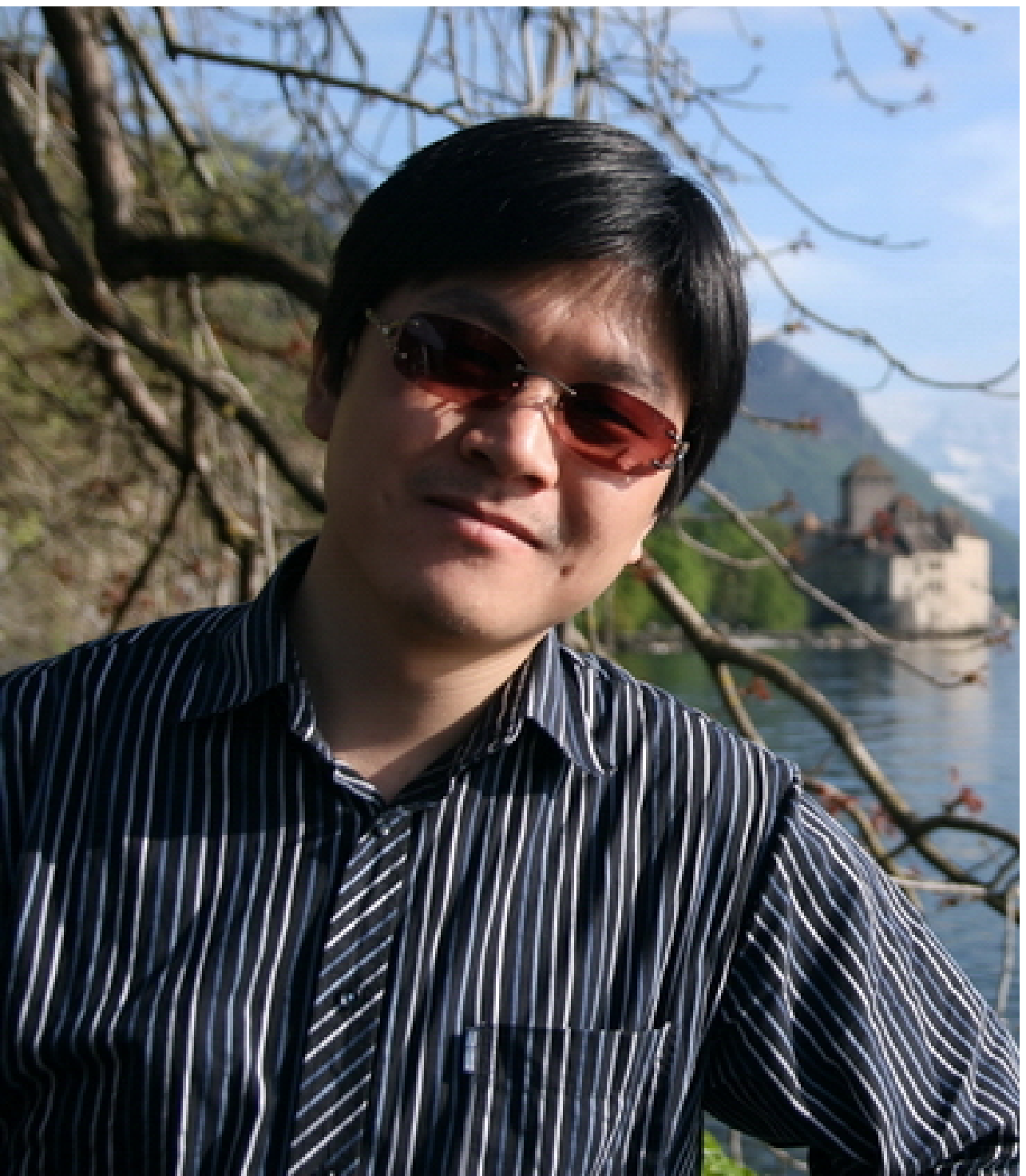}}]{Jeong-woo Cho}
received his B.S., M.S., and Ph.D. degrees in Electrical Engineering and Computer Science from KAIST, Daejeon, South Korea, in 2000, 2002, and 2005, respectively. From September 2005 to July 2007, he was with the Telecommunication R\&D Center, Samsung Electronics, South Korea, as a Senior Engineer. From August 2007 to August 2010, he held postdoc positions in the School of Computer and Communication Sciences, \'Ecole Polytechnique F\'ed\'erale de Lausanne (EPFL), Switzerland, and at the Centre for Quantifiable Quality of Service in Communication Systems, Norwegian University of Science and Technology (NTNU), Trondheim, Norway. He is now an assistant professor in the School of Information and Communication Technology at KTH Royal Institute of Technology, Stockholm, Sweden. His current research interests include performance evaluation in various networks such as peer-to-peer network, wireless local area network, and delay-tolerant network.
\end{IEEEbiography}

\begin{IEEEbiography}[{\includegraphics[width=1in,height=1.25in,keepaspectratio]{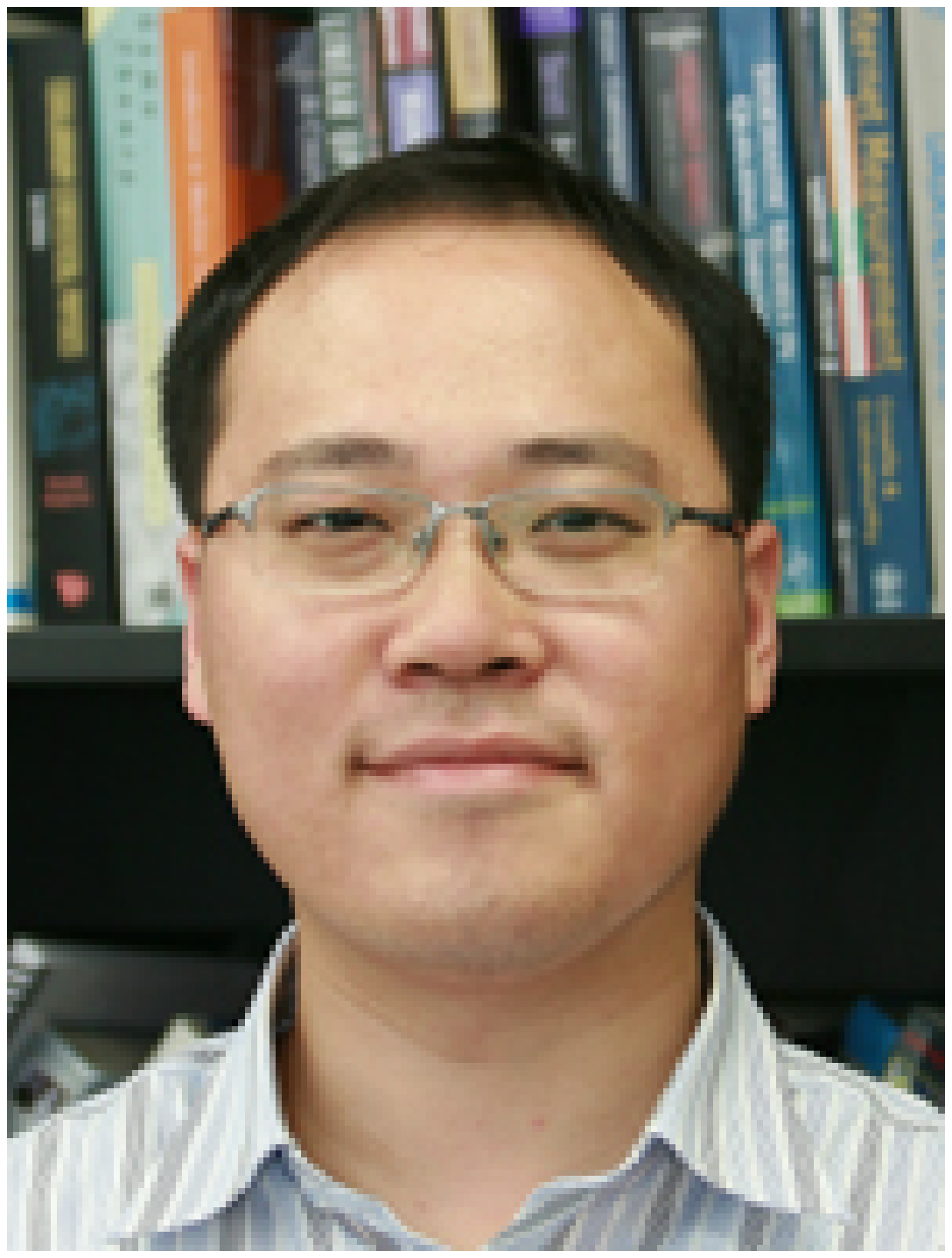}}]{Yung Yi}
received his B.S. and the M.S. in the School of Computer Science and Engineering from Seoul National University, South Korea in 1997 and 1999, respectively, and his Ph.D. in the Department of Electrical and Computer Engineering at the University of Texas at Austin in 2006. From 2006 to 2008, he was a post-doctoral research associate in the Department of Electrical Engineering at Princeton University. He is now an associate professor at the Department of Electrical Engineering at KAIST, South Korea. He has been serving as a TPC member at various conferences including ACM Mobihoc, Wicon, WiOpt, IEEE Infocom, ICC, Globecom, ITC, and WASA. His academic service also includes the local arrangement chair of WiOpt 2009 and CFI 2010, the networking area track chair of TENCON 2010, the publication chair of CFI 2011-2012, and the symposium chair of a green computing, networking, and communication area of ICNC 2012. He has served as a guest editor of the special issue on Green Networking and Communication Systems of IEEE Surveys and Tutorials, an associate editor of Elsevier Computer Communications Journal, and an associate editor of Journal of Communications and Networks. He has also served as the co-chair of the Green Multimedia Communication Interest Group of the IEEE Multimedia Communication Technical Committee. His current research interests include the design and analysis of computer networking and wireless systems, especially congestion control, scheduling, and interference management, with applications in wireless ad hoc networks, broadband access networks, economic aspects of communication networks, and greening of network systems.
\end{IEEEbiography}

\end{document}